\documentclass[letterpaper,11pt]{article}

\newcommand{\TitleMetadataForHyperref}{Faster Distributed Δ-Coloring via a Reduction to MIS}
\newcommand{\AuthorMetadataForHyperref}{Yann Bourreau, Sebastian Brandt, Alexandre Nolin}

\usepackage{1a_arxiv_v2_package_loading}

\usepackage{1b_arxiv_v2_paper_macros}

\begin{document}
\setcounter{page}{0}
\thispagestyle{empty}

\begin{mycover}
    {\huge\bfseries\boldmath Faster Distributed $\Delta$-Coloring via a Reduction to MIS\par}
    \bigskip
    \bigskip
    \bigskip
    \textbf{Yann Bourreau}
    \myorcid{0009-0001-1819-8348}
    \myaff{CISPA Helmholtz Center for Information Security}

    \textbf{Sebastian Brandt}
    \myorcid{0000-0001-5393-6636}
    \myaff{CISPA Helmholtz Center for Information Security}
    
    \textbf{Alexandre Nolin}
    \myorcid{0000-0002-3952-0586}
    \myaff{CISPA Helmholtz Center for Information Security}
\end{mycover}
\bigskip

\begin{abstract}
    Recent improvements on the deterministic complexities of fundamental graph problems in the LOCAL model of distributed computing have yielded state-of-the-art upper bounds of $\Otilde(\log^{5/3} n)$ rounds for maximal independent set (MIS) and $(\Delta + 1)$-coloring [Ghaffari, Grunau, FOCS'24] and $\Otilde(\log^{19/9} n)$ rounds for the more restrictive $\Delta$-coloring problem [Ghaffari, Kuhn, FOCS'21; Ghaffari, Grunau, FOCS'24; Bourreau, Brandt, Nolin, STOC'25].
    In our work, we show that $\Delta$-coloring can be solved deterministically in $\Otilde(\log^{5/3} n)$ rounds as well, matching the currently best bound for $(\Delta + 1)$-coloring.

    We achieve our result by developing a reduction from $\Delta$-coloring to MIS that guarantees that the (asymptotic) complexity of $\Delta$-coloring is at most the complexity of MIS, unless MIS can be solved in sublogarithmic time, in which case, due to the $\Omega(\log n)$-round $\Delta$-coloring lower bound from [BFHKLRSU, STOC'16], our reduction implies a tight complexity of $\Theta(\log n)$ for $\Delta$-coloring.
    In particular, any improvement on the complexity of the MIS problem will yield the same improvement for the complexity of $\Delta$-coloring (up to the true complexity of $\Delta$-coloring).

    Our reduction also yields improvements for $\Delta$-coloring in the randomized LOCAL model and when complexities are parameterized by both $n$ and $\Delta$.
    For instance, we obtain a randomized complexity bound of $\Otilde(\log^{5/3} \log n)$ rounds (improving over the state of the art of $\Otilde(\log^{8/3} \log n)$ rounds) on general graphs and tight complexities of $\Theta(\log n)$ and $\Theta(\log \log n)$ for the deterministic, resp.\ randomized, complexity on bounded-degree graphs.
    In the special case of graphs of constant clique number (which for instance include bipartite graphs), we additionally give a reduction to the $(\Delta+1)$-coloring problem.
\end{abstract}

\clearpage

\setcounter{page}{0}
\thispagestyle{empty}
\tableofcontents

\clearpage

\section{Introduction}
In the last decade, the theory of distributed graph algorithms has seen tremendous progress, leading to faster and faster algorithms for many of the problems central to the area.
Much of this progress has been achieved in the context of so-called \emph{locally checkable problems\footnote{A locally checkable problem is a problem where the correctness of the solution is defined via local constraints. A claimed global solution is correct if and only if it satisfies the local constraints at each node (or edge).}}---a problem class that has been in the center of attention since the beginning of the field and contains most of the fundamental problems studied throughout the years, such as maximal independent set (MIS) or coloring problems.
A particular focus of the works involved in this effort has been on the complexity of the aforementioned problems in the standard \emph{LOCAL model} of distributed computing, introduced by Linial in the 80s~\cite{linial1987distributive}.

\paragraph{The LOCAL Model.}
In the LOCAL model, each node of the input graph is considered as a computational entity and the edges of the input graph model communication channels via which the nodes can communicate.
More precisely, computation proceeds in synchronous rounds, and in each round each node first sends an arbitrarily large message to each of its neighbors, then receives the messages sent by its neighbors, and finally performs some (arbitrarily powerful) internal computation on the knowledge is has up to this point (e.g., to determine which message to send in the next round).
Each node executes the same algorithm and has to decide at some point to select its output, upon which it terminates and does not participate in sending messages anymore.
To solve a given problem correctly, the collection of all the local outputs chosen by the nodes needs to satisfy the specifications of the respectively considered problem, e.g., if the problem is to compute a proper coloring with a specified set of colors, then the global output is correct if the output of each node is a color from the specified set and, for any two neighboring nodes, the colors chosen by those nodes are different.

In the beginning, each node $v$ is only aware of the information associated with itself---its degree $\deg(v)$, an internal numbering of its incident edges from $1$ to $\deg(v)$ called a \emph{port numbering}, some \emph{symmetry-breaking information} specified in the following, and possibly some problem-specific input such as a list of colors in a list coloring problem---and some standard information about graph parameters such as the values of the number $n$ of nodes and the maximum degree $\Delta$ of the input graph.
The mentioned symmetry-breaking information is what distinguishes the deterministic and randomized LOCAL model: in the deterministic LOCAL model, each node is equipped with a numerical identifier from the range $\{ 1, \dots, \poly(n) \}$ that is guaranteed to be unique over the whole graph, whereas in the randomized LOCAL model, each node is instead provided with a private random bit string of arbitrary length.

The complexity of a distributed algorithm (considered as a function of $n$) is the worst-case number of rounds until the algorithm terminates taken over all $n$-node graphs (and, in the case of deterministic algorithms, over all assignments of identifiers).
The complexity of a problem is the complexity of an optimal algorithm for that problem.
In the case of randomized algorithms, the considered algorithms will be Monte-Carlo algorithms that we require to provide a correct output with high probability (w.h.p.), i.e., with probability at least $1 - 1/n$.

\paragraph{The Big $4$.}
A group of problems whose state-of-the-art complexity bounds in the LOCAL model have undergone an especially noteworthy development in the last 8 years are what is colloquially referred to as the ``Big 4''---MIS, maximal matching, $(\Delta + 1)$-coloring, and $(2\Delta - 1)$-edge coloring---and closely related problems such as ruling set problems.
These problems are greedy-type problems, i.e., problems that can be solved sequentially by an algorithm that simply iterates through all nodes or edges in an arbitrary order and assigns an output to the respectively processed node/edge in a local greedy fashion that only takes a constant-hop neighborhood of the node/edge (and the already assigned outputs in that neighborhood) into account, which makes them amenable to a variety of algorithmic techniques.
In particular two of these techniques---network decompositions and rounding---have led to a recent flurry of papers~\cite{fischer2020improved,rozhon2020polylogarithmic,ghaffari2021improved,ghaffari2021deterministic,faour2023local,ghaffari2023faster,ghaffari2023improveddnd} culminating in the current state-of-the-art deterministic complexity upper bound of $\Otilde(\log^{5/3} n)$ rounds\footnote{We will use the notation $\Otilde(\cdot)$ to hide a factor that is polylogarithmic in the argument.} \cite{GG_focs24} for all of the Big 4.
For two of those problems---maximal matching and MIS---a lower bound of $\Omega(\log n/\log \log n)$ rounds~\cite{balliu2021lower} has been obtained recently whereas for the two coloring problems nothing better than the decades-old lower bound of $\Omega(\log^* n)$ rounds by~\cite{linial1992locality} is known.

\paragraph{Beyond greedy.}
Besides the aforementioned four (and other important) greedily solvable problems there are a number of fundamental locally checkable problems studied extensively in the last years that do not admit such a sequential greedy process, with two of the most prominent examples being $\Delta$-coloring and $(\Delta + 1)$-edge coloring (as well as other sub-$(2\Delta - 1)$-edge colorings).
This type of problems seems to be much less amenable to generic techniques such as the aforementioned network decomposition and rounding, with the current deterministic state-of-the-art complexity upper bounds being moderately to substantially weaker than what is known for the Big 4: for instance, for $\Delta$-coloring, as observed in~\cite[Theorem A.2, arxiv version]{bourreau2025faster} an $\Otilde(\log^{19/9} n)$-round algorithm can be obtained by combining~\cite{ghaffari2021deterministic} and~\cite{GG_focs24}, while for $(\Delta + 1)$-edge coloring
no $O(\poly(\log n))$-round algorithm is known to exist.
This raises the following fundamental question.

\vspace{0.05cm}
\begin{tcolorbox}
	\begin{oq}\label{oq1}
    How can we make use of the toolbox applicable to the Big 4 also for non-greedy problems? Can we reduce (some) non-greedy problems to (one of) the Big 4 in some fashion?
    \end{oq}
\end{tcolorbox}
\vspace{0.05cm}

Obtaining a reduction without overhead from some non-greedy problem like $\Delta$-coloring or $(\Delta + 1)$-edge coloring to one of the Big 4 would be excellent news not only because this would immediately improve the best known upper bound for the considered problems
but also due to the amenability of the Big 4 to generic techniques, which likely makes further improvements on the state-of-the-art upper bounds easier to achieve (which would then immediately transfer to the considered non-greedy problem).
Unfortunately, such a reduction cannot exist if we require the reduction to be at least remotely degree-preserving: on constant-degree graphs, all of the big 4 can be solved in $O(\log^* n)$ rounds~\cite{panconesi01simple} whereas $\Delta$-coloring and any sub-$(2\Delta - 1)$-edge coloring require $\Omega(\log n)$ rounds~\cite{chang2019exponential,chang2019distributed}.\footnote{More generally, there exist $O(\log^* n + \Delta)$-round algorithms for all of the Big 4~\cite{panconesi01simple,barenboim14distributed} and $\Omega(\log_{\Delta} n)$-round lower bounds for $\Delta$-coloring and any sub-$(2\Delta - 1)$-edge coloring~\cite{chang2019exponential,chang2019distributed}.}

In this work, we focus on $\Delta$-coloring.
In particular, we will show how to circumvent the aforementioned impossibility and obtain the reduction desired in~\cref{oq1} (without a substantial increase in the maximum degree of the input graph), which provides the first such reduction for any of the classic non-greedy graph problems studied across different models for decades.

\paragraph{$\Delta$-coloring.}
In the $\Delta$-coloring problem, the task is to color the nodes of the input graph with colors from $\{ 1, \dots, \Delta \}$ such that no two adjacent nodes receive the same color.
By Brooks' Theorem~\cite{brooks1941colouring}, such a coloring always exists for a connected graph unless it is an odd cycle or a $(\Delta+1)$-clique.
We will use the standard assumption implicit when studying $\Delta$-coloring from an algorithmic perspective that only graphs that do not fall into the small class of unsolvable graphs are allowed as input instances.
Moreover, throughout the paper, we will also assume that $\Delta \geq 3$, which is another standard assumption due to the fact that the behavior of $\Delta$-coloring is both quite different and well-understood in the case of $\Delta \leq 2$.

While $\Delta$-coloring is an important problem across many models, it has turned out to be an exceptionally influential problem in the LOCAL model.
From the perspective of lower bounds, it was instrumental in the development of entire lower bound techniques, including the \emph{round elimination technique}~\cite{brandt2016lower,brandt2019automatic} that has been used in an abundance of works in the last decade to prove lower bounds for many of the most important problems studied in the LOCAL model (see, e.g., \cite{brandt2016lower,balliu2021lower,balliu2021improved,balliu2022ruling,balliu2022distributed,balliu2023distributed}).
In particular, $\Delta$-coloring was one of the problems for which lower bounds were proved in the initial publications for both round elimination~\cite{brandt2016lower} and another recent lower bound technique called \emph{Marks' technique}~\cite{DetMarks,Marks_Coloring}.

Building on the randomized $\Omega(\log \log n)$-round lower bound for $\Delta$-coloring proved in~\cite{brandt2016lower}, Chang, Kopelowitz, and Pettie proved asymptotically tight bounds for $\Delta$-coloring on trees of $\Theta(\log_{\Delta} \log n)$ rounds randomized and $\Theta(\log_{\Delta} n)$ rounds deterministically (for any $\Delta \geq 55$) \cite{chang2019exponential}.
This result not only provided the first exponential separation between the randomized and deterministic complexities of a locally checkable problem but also identified $\Delta$-coloring as the first problem of provably ``intermediate'' complexity initiating a long line of research now known under the name \emph{distributed complexity theory (of LCL problems\footnote{An LCL problem is essentially a locally checkable problem studied on graphs of constant maximum degree.})}, which essentially mapped out the complexity landscape of locally checkable problems on constant-degree graphs.

The current state-of-the-art upper bounds for $\Delta$-coloring on general ($\Delta$-colorable, $\Delta \geq 3$) graphs are the culmination of three decades of insights, starting with Panconesi and Srinivasan~\cite{panconesi95delta} who proved that recoloring a path of only $O(\log_\Delta n)$ nodes suffices to extend a partial $\Delta$-coloring to an additional node.
Ghaffari, Hirvonen, Kuhn, and Maus~\cite{GHKM_dc21} showed that $\Delta$-coloring admits much faster algorithms in the randomized setting, as allowing for a small error probability makes it possible to compute a partial coloring with the property that extending the coloring to an additional node never involves recoloring more than an
$O(\log_\Delta \log n)$-sized path.
These two results~\cite{panconesi95delta,GHKM_dc21} with the current state-of-the-art for other problems (degree+1-list coloring, ruling sets) suffice to obtain upper bounds of $\poly(\log n)$ and $\poly(\log \log n)$ on the deterministic and randomized complexities of $\Delta$-coloring.
$\Delta$-coloring was shown to even admit an $O(\log^* n)$ randomized algorithm by Fischer, Halldórsson, and Maus~\cite{fischer2023fast}, when $\Delta$ is large enough ($\Delta \in \omega(\log^{21} n)$).
The current deterministic state-of-the-art is due to Bourreau, Brandt, and Nolin~\cite{bourreau2025faster}, together with the current state-of-the-art bounds for ruling sets and coloring graphs with a special ``layered'' structure~\cite{ghaffari2021deterministic}.
Two very recent papers obtained improved bounds in some special graph classes~\cite{JM_podcba25,JMS_arxiv25}.

We refer the reader to~\cite{bourreau2025faster} for an in-depth discussion of the importance of $\Delta$-coloring, which also outlines its relevance for future developments in the LOCAL model and its close relation to a number of fundamental open problems.
As such, obtaining better bounds for $\Delta$-coloring is one of the major objectives of current research in the LOCAL model.

\subsection{Our Contributions}
To describe our contributions in full generality, we will use the following notation to abbreviate the complexities of important problems.

For parameters $n, \Delta$, we denote the deterministic complexities of $\Delta$-coloring, $(\Delta+1)$-coloring, and MIS\footnote{In the MIS problem, the task is to select a subset of the nodes (called an MIS) such that no two selected nodes are adjacent and any node that is not selected has a selected neighbor.} on $n$-node graphs of maximum degree $\Delta$ by $\Tdelc(n,\Delta)$, $\Tdelpoc(n,\Delta)$, and $\Tmis(n,\Delta)$, respectively.
For parameters $n, k$, we denote the deterministic complexity of a problem called hypergraph sinkless orientation (HSO)\footnote{For a formal definition of hypergraph sinkless orientation, see~\cref{sec:prelims}.} on $n$-node graphs with \emph{degree-to-rank ratio} $\delta/r \geq k$ by $\Thso(n,k)$, where $\delta$ denotes the minimum degree of the input graph and $r$ the maximum rank over all hyperedges.
To indicate randomized complexities, we will add ``$\text{rand}$'' in the index (and add $\delta$ as a third parameter to the complexity of HSO since the state-of-the-art runtime for that problem depends on $\delta$ in the randomized setting), using the expressions $\Tdelcrand(n,\Delta)$, $\Tmisrand(n,\Delta)$, and $\Thsorand(n,k,\delta)$.
For technical reasons, for the definition of $\Thsorand(n,k,\delta)$ we require all input instances for hypergraph sinkless orientation to satisfy $\delta > 320r \log r$.

All of the above expressions are to be understood as representing the (unknown) \emph{tight} complexities of the considered problems.

\paragraph{Deterministic complexities.}
As our main result, we show how to deterministically reduce $\Delta$-coloring to MIS and hypergraph sinkless orientation.

\begin{restatable}{theorem}{ThmLOCALmain}
    \label{thm:local-main}
    Deterministic $\Delta$-coloring reduces to solving $O(1)$ instances of deterministic MIS on graphs of maximum degree $O(\poly(\Delta))$ and at most $n$ nodes, and solving an instance of deterministic hypergraph sinkless orientation with degree-to-rank ratio $\delta/r \in \Omega(\poly(\Delta))$.
    In particular,
    \begin{equation*}
    \Tdelc(n,\Delta) \in O(\Tmis(n,\poly(\Delta)) + \Thso(n,\poly(\Delta)))\ .
\end{equation*}
\end{restatable}

Note that hypergraph sinkless orientation is a problem whose difficulty decreases with growing degree-to-rank ratio, which is why we may assume that the second parameter in the second summand is $\poly(\Delta)$ even though it may be in $\omega(\poly(\Delta))$.

As shown in~\cite{BMNSU_halls_thm_soda25}, hypergraph sinkless orientation can be solved deterministically in $O(\log_{\delta/r} n)$ rounds, which, combined with the condition $\delta/r \in \Omega(\poly(\Delta))$ from~\cref{thm:local-main}, implies that the instance of hypergraph sinkless orientation mentioned in~\cref{thm:local-main} can be solved in $O(\log_{\Delta} n)$ rounds.
As $\Delta$-coloring has a deterministic lower bound of $\Omega(\log_{\Delta} n)$ rounds~\cite{chang2019exponential}, \cref{thm:local-main} therefore yields a clean reduction from $\Delta$-coloring to MIS that shows that the complexity of $\Delta$-coloring (as a function of $n$) is at most the complexity of MIS, unless the complexity of MIS turns out to be in $o(\log n)$ in which case the reduction yields a tight complexity of $\Theta(\log n)$ for $\Delta$-coloring.
This answers~\cref{oq1} in a strong sense, without a substantial increase in the maximum degree in the reduction, yielding all the discussed benefits.

In particular, with the (differently parameterized) state-of-the-art bounds of $\Otilde(\log^{5/3} n)$~\cite{GG_focs24}, $O(\log^2 \Delta \cdot \log n)$~\cite{faour2023local}, and $O(\Delta + \log^* n)$~\cite{barenboim14distributed} rounds for the deterministic complexity of MIS,
we obtain the following new bounds for the complexity of $\Delta$-coloring.

\begin{restatable}{corollary}{CorLOCALmain}
    \label{cor:local-main}
    $\Delta$-coloring admits deterministic LOCAL algorithms of the following complexities:
    \begin{enumerate}
        \item $\Otilde(\log^{5/3} n)$,
        \item $O(\log^2 \Delta \cdot \log n)$,
        \item $O(\poly(\Delta) + \log_\Delta n)$.
    \end{enumerate}
\end{restatable}

The results in \Cref{cor:local-main} yields a substantial improvement over the previous state of the art of $\Otilde(\log^{19/9} n)$~\cite{ghaffari2021deterministic,GG_focs24,bourreau2025faster} and $O(\log^4 \Delta + \log^2 \Delta \cdot \log n \log^* n)$~\cite{bourreau2025faster} rounds.
Our upper bound of $\Otilde(\log^{5/3} n)$ rounds for $\Delta$-coloring matches the state-of-the-art upper bound for the complexity of the $(\Delta + 1)$-coloring problem, implying that improvements for $(\Delta + 1)$-coloring are necessary to improve our bound for $\Delta$-coloring further. 
\Cref{cor:local-main} also provides the first algorithm of the form $O(f(\Delta) + \log_{\Delta} n)$ for any function $f$, which in particular implies that the tight complexity of $\Delta$-coloring on constant-degree graphs is $\Theta(\log n)$.
This closes the gap left by the best previous upper bound on constant-degree graphs, $O(\log n \log^* n)$~\cite{bourreau2025faster}.

Moreover, due to~\cref{thm:local-main}, any further improvements on the complexity of MIS will yield improvements on the complexity of $\Delta$-coloring (up to the optimal complexity of $\Delta$-coloring).
We note that the best currently known deterministic lower bound for MIS stands at $\Omega(\min\{ \Delta + \log^* n, \log_{\Delta} n \})$ rounds~\cite{linial1992locality,balliu2021lower}.
With the last years seeing essentially yearly improvements in the state-of-the-art upper bound for MIS~\cite{rozhon2020polylogarithmic,ghaffari2021improved,faour2023local,ghaffari2023faster,GG_focs24} but no change in the aforementioned lower bound, it is plausible to assume that the lower bound might be tight.
In this case, the study of the complexity of $\Delta$-coloring can be replaced entirely by studying the MIS problem, which might be more approachable, evidenced by the aforementioned improvements.
The same holds in any scenario where the complexity of MIS is at most the complexity of $\Delta$-coloring, which is plausible to assume since, as observed above, even the less restrictive $(\Delta + 1)$-coloring problem currently stands at the same complexity of $\Otilde(\log^{5/3} n)$ rounds as MIS.

\paragraph{Randomized complexities.}

Our technique can also be used in the randomized LOCAL model yielding a statement analogous to~\cref{thm:local-main}.

\begin{restatable}{theorem}{ThmLOCALmainrand}
    \label{thm:local-main-rand}
    Randomized $\Delta$-coloring reduces to solving $O(1)$ instances of randomized MIS on graphs of maximum degree $O(\poly(\Delta))$ and at most $n$ nodes, and solving an instance of randomized hypergraph sinkless orientation with degree-to-rank ratio $\delta/r \in \Omega(\poly(\Delta))$, $\delta  > 320r \log r$, and $\delta \in O(\poly(\Delta))$.
    In particular,
    \begin{equation*}
    \Tdelcrand(n,\Delta) \in O(\Tmisrand(n,\poly(\Delta)) + \Thsorand(n,\poly(\Delta),\poly(\Delta)))\ .
\end{equation*}
\end{restatable}

We emphasize that~\cref{thm:local-main-rand} is not a direct corollary of~\cref{thm:local-main} due to the extra condition $\delta > 320r \log r$.
This condition is required for the application of the state-of-the-art algorithm for randomized hypergraph sinkless orientation, which has a runtime of $O(\log_{\delta/r} \delta + \log_{\delta/r} \log n)$ rounds~\cite{BMNSU_halls_thm_soda25}.

Using that $\Delta$-coloring admits an $O(\log^* n)$-round randomized algorithm when $\Delta \in \omega(\log^{21} n)$~\cite{fischer2023fast}, and given the current state-of-the-art bounds of $O(\log \Delta) + \Otilde(\log^{2} \log n)$ \cite{ghaffari16improved,GG_focs24}, $\Otilde(\log^{5/3} \Delta) + \Otilde(\log^{5/3} \log n)$ \cite{ghaffari16improved,GG_focs24}, $O(\log^{3} \Delta + \log^{2} \Delta \cdot \log \log n)$ \cite{ghaffari16improved,faour2023local}, and 
$O(\Delta + \log^* n)$~\cite{barenboim14distributed} rounds for the randomized complexity of MIS, we obtain the following new bounds for the randomized complexity of $\Delta$-coloring.

\begin{restatable}{corollary}{CorLOCALmainRAND}
    \label{cor:local-main-rand}
    $\Delta$-coloring admits randomized LOCAL algorithms of the following complexities w.h.p.:
    \begin{enumerate}
        \item $\Otilde(\log^{5/3} \log n)$,
        \item $O(\log^2 \Delta \cdot \log \log n)$,
        \item $O(\poly(\Delta) + \log_\Delta \log n)$.
    \end{enumerate}
\end{restatable}

\Cref{cor:local-main-rand} improves over the previous state of the art of $\Otilde(\log^{8/3} \log n)$ rounds by~\cite{bourreau2025faster}, coming much closer to the lower bound of $\Omega(\log \log n)$ from~\cite{brandt2016lower}.
Moreover, it implies that the complexity of $\Delta$-coloring on constant-degree graphs is $\Theta(\log \log n)$.
This provides further evidence for the correctness of the Chang-Pettie Conjecture~\cite{ChangP19} that can be stated as follows: any LCL problem\footnote{Recall that LCL problems are defined only on constant-degree graphs.} with (asymptotically) sublogarithmic randomized complexity can be solved in $O(\log \log n)$ rounds.
With $\Delta$-coloring, \cref{cor:local-main-rand} removes one of the most natural candidates for refuting the conjecture from the list of obstacles.

\paragraph{LOCAL complexities for important graph classes.}
We also provide an algorithm for the case of graphs of bounded clique number.

\begin{restatable}{theorem}{ThmLOCALcliqueNum}
    \label{thm:local-clique-number}
    Let $\omega \geq 3$ be an integer, and assume $\Delta \geq \omega$. In both deterministic and randomized \LOCAL, $\Delta$-coloring $K_{\omega+1}$-free graphs reduces to solving $(\Delta+1)$-coloring, applying $\omega$ rounds of color reduction, and solving $\omega$-coloring on a graph of maximum degree $\omega$. 
\begin{align*}
    \Tdelc(n,\Delta) 
    & \in O(\Tdelpoc(n,\Delta) + \omega + \Tdelc(n,\omega))
    & \text{(on $K_{\omega+1}$-free graphs)}\ ,
    \\
    \Tdelcrand(n,\Delta)
    & \in O(\Tdelpocrand(n,\Delta) + \omega + \Tdelcrand(n,\omega))
    & \text{(on $K_{\omega+1}$-free graphs)}\ .
    \end{align*}
\end{restatable}

In particular, this implies that triangle-free graphs (which include bipartite graphs) can be $\Delta$-colored in $O(\sqrt{\Delta\log \Delta} + \log n)$ rounds deterministically, and $O(\sqrt{\Delta\log \Delta} + \log \log n)$ rounds with high probability~\cite{MausTonoyan_dc22}. With the current state-of-the-art for MIS, these last runtimes beat all the complexities that follow from \cref{thm:local-main} when $\Delta \in \Otilde(\log^2 n)$ in the deterministic setting, and when $\Delta \in \Otilde(\log^{10/3}  \log n)$ in the randomized setting.

\subsection{Overview of Our Techniques}

To solve the $\Delta$-coloring problem in the deterministic LOCAL model, our approach relies on creating nodes with some guaranteed flexibility in their coloring throughout the graph. Referred to as $T$-nodes, these vertices have two of their neighbors assigned the color $1$, a color that remains fixed throughout the algorithm. This guaranteed repetition of the color $1$ in the neighborhood of $T$-nodes guarantees that a partial $\Delta$-coloring can always be extended to them,
regardless of the colors assigned to their other neighbors.

The core objective is to select $T$-nodes such that every node in the graph has a short uncolored path to a $T$-node, i.e., a path without a node permanently colored with $1$ at the beginning of the algorithm. Once such a structure is in place, the graph can be colored in layers, as follows. Each node in the graph is assigned to a layer according to the length of the shortest uncolored path to a $T$-node. Specifically, a node belongs to layer $i$ if its closest $T$-node is at distance $i$ along an uncolored path. Coloring is done in descending order of layers: nodes in the highest-numbered layer are colored first, followed by nodes in progressively lower layers.
For any $i\geq 1$, a partial $\Delta$-coloring of the layers numbered $i+1$ and above can always be extended to nodes in the layer $i$ due to the fact that each of them has at least one uncolored neighbor in layer $i-1$, giving it some flexibility. 
Layer $0$ consists only of the $T$-nodes, which similarly can always be colored using colors from the set $\set{2,\dots,\Delta}$ due to their two neighbors that have been (permanently) colored with $1$ early in the algorithm.

Constructing such a set of $T$-nodes is inherently a balancing act. On the one hand, ensuring that each node in the graph has a nearby $T$-node suggests adding as many $T$-nodes as possible. On the other hand, each $T$-node also disconnects potential uncolored paths as two of its neighbors get colored with color $1$.
An especially bad situation to avoid is one in which the nodes of color $1$ disconnect the graph and create an uncolored connected component isolated from all $T$-nodes.

To address this challenge, we construct a carefully designed cluster partition of the original graph, and define a multihypergraph from this cluster partition. Nodes of this instance correspond to clusters of nodes in the original graph, and are partitioned into two sets, which we call \emph{flex} and \emph{link}.
Adjacency between two virtual nodes represents adjacency of their respective clusters in the original graph. 
The flex clusters are responsible for creating $T$-nodes, while the link clusters ensure connectivity: they guarantee that every node in the graph is connected to some $T$-node through an uncolored path of \emph{constant} length.
In addition, some parts of the graph are not conducive to the creation of $T$-nodes due to a lack of local expansion. We deal with these parts of the graph through other arguments, finding in them small induced subgraphs (degree-choosable components, abbreviated DCCs) to which extending a partial coloring is always possible and simple, regardless of the coloring decisions made outside them in the graph.

As
each flex cluster creates a $T$-node in itself and permanently colors two of the $T$-node's neighbors, some connections are lost between nodes from the flex cluster and nodes in adjacent link clusters.
Our goal is to guarantee that each link cluster keeps at least one uncolored path to the ``center'' of a neighboring flex cluster.
To achieve this, we use a recent algorithmic extension of the sinkless orientation problem, the hypergraph sinkless orientation algorithm, which we apply to the hypergraph built from our cluster partition.
In the instance we build,
grabbing an hyperedge on the flex clusters' side corresponds to the creation of a particular $T$-node, while a grabbed hyperedge on the link nodes' side, by preventing the creation of a specific $T$-node, maintains connectivity from the link cluster to the center of the flex cluster.
Solving HSO ensures that each flex cluster creates at least one $T$-node in itself, and that every link cluster has at least one short viable path to a flex cluster's created $T$-node, a path that avoids permanently colored neighbors of $T$-nodes.
Computing the clusters, breaking symmetry between DCCs, and coloring nodes layer by layer, all reduce to solving instances of the Maximal Independent Set problem (MIS).
In the randomized setting, similar results can be obtained by employing randomized variants of the algorithmic subroutines used in the deterministic version of our algorithm.

In general graphs, when given a coloring and considering the subgraph induced by $k \leq \Delta$ color classes, we cannot exclude the possibility that the induced subgraph contains $k$-cliques, making it impossible to reduce the number of colors used in that subgraph.
Assuming a graph to be free of cliques of size $k < \Delta+1$ removes this obstacle, making it possible to solve $\Delta$-coloring 
by first $(\Delta+1)$-coloring the graph, and after some additional preprocessing, reducing to $k-1$ the number of colors used in the subgraph induced by $k$ colors.

\section{Preliminaries}\label{sec:prelims}

In this section, we collect some definitions and results that will be useful throughout the paper.
For a set $S$, we denote by $2^S$ its power set, i.e., the set of all subsets of $S$. The $\log^*$ function is defined by $\forall x \in (-\infty,1], \log^*(x) = 0$ and $\forall x > 1, \log^*(x) = 1 + \log^*(\log(x))$.

\paragraph{Graph and hypergraph notation}
A graph $G=(V,E)$ consists of a set of nodes $V$ and a set of edges $E$ where each edge $e \in E$ corresponds to a set $V_e \subseteq V$ containing exactly two distinct nodes, that is, $\card{V_e} = 2$.
$G' = (V',E')$ is a subgraph of $G = (V,E)$ iff $V' \subseteq V$ and $E' \subseteq E$. 
For a graph $G=(V,E)$ and a set $U \subseteq V$, the subgraph induced by $U$ in $G$ is the graph $G[U] = (U,E[U])$ where $E[U] = \set{uv \in E: u\in U,v\in U}$. Given a graph $G$, we denote by $V(G)$ and $E(G)$ its sets of nodes and edges.
Throughout the paper, $\Delta \geq 3$ denotes the maximum degree of the given input graph, and $n$ denotes its number of nodes. Two subgraphs $H$ and $H'$ of $G$ are independent if for any nodes $v\in H$ and $v'\in H'$, $v$ and $v'$ are not adjacent.

Hypergraphs are a generalization of graphs where the edges, now called hyperedges, can contain any nonnegative number of nodes. For an integer $r\geq 1$, a hypergraph $H=(V,E)$ is said to have rank $r$ if for each hyperedge $e \in E$, $\card{V_e} \in [1,r]$. Multigraphs (resp.\ multihypergraphs) generalize graphs (resp.\ hypergraphs) by allowing each edge (resp.\ hyperedge) to exist with multiplicity. That is, there can exist two distinct edges $e,e' \in E$ such that their corresponding sets of nodes are equal, $V_e = V_{e'} \subseteq V$. For brevity, we omit the ``multi'' prefix in most of the paper. All the objects we consider in this paper fall in one of two categories: graphs without multiplicity, or hypergraphs with multiplicity.

The degree $\deg(v)$ of a node $v$ in all these definitions is the number of edges containing said node: $\deg(v) = \card{\set{e \in E: v \in V_e}}$.
For a node $v \in V$, its neighborhood is $N(v) = (\bigcup_{e \in E : v\in V_e} V_e) \setminus \set{v}$.
Similarly, the neighborhood of a set of nodes $U \subseteq V$ is $N(U) = (\bigcup_{e \in E : v\in V_e} V_e) \setminus U$, so $N(\set{v}) = N(v)$.

A path between two nodes $v,v' \in V$ is a sequence of nodes $v_0,v_1,\dots,v_k$ such that $v_0 = v$, $v_k = v'$, and for each $i\in \set{0,\dots,k-1}$, $\exists e \in E$, $\set{v_i,v_{i+1}} \subseteq V_e$. $k$ is the length of the path. The smallest integer $k$ such that there exists a path of length $k$ between two nodes $v,v'$ is their distance, $\dist(v,v')$. When no such path exists, $\dist(v,v') = +\infty$. The distance between a set of nodes $U \subseteq V$ and a node $v \in V$ is defined as $\dist(v,U) = \min_{u \in U} \dist(v,u)$. Note that while distance between nodes verifies the triangle inequality ($\dist(v,v'') \leq \dist(v,v') + \dist(v',v'')$), this is not necessarily true with sets (e.g., if $v \neq v'$, $\dist(v,v') > \dist(v,\set{v,v'}) + \dist(v',\set{v,v'}) =0$).
For a positive integer $k>0$ and graph $G=(V,E)$, the power graph of $G$ of order $k$ is $G^k = (V,E_k)$, whose set of edges is $E_k = \set{uv \mid \dist(u,v) \in [1,k]}$.
We define the ball of radius $r$ centered on $v$ as $B(v,r) = \set{u \in V \mid \dist(u,v) \leq r}$, and the sphere of radius $r$ centered on $v$ as $S(v,r) = \set{u \in V \mid \dist(u,v) = r}$.

When not clear from context and the distinction matters, we add subscripts to $\deg$, $\dist$, $B$, and $S$, to indicate in which graph the degrees and distances are measured (e.g., $\deg_{G}(v)$, $\dist_{G'}(v,v')$, $B_{H}(v,r)$). A subgraph $H$ of $S\in G$ is of weak diameter $r$ if the diameter of $H$ is of $r$ in $G$. A subgraph $H$ of $S\in G$ is of strong diameter $r$ if the diameter of $H$ is of $r$ in $S$.

\paragraph{Hypergraph sinkless orientation}
Sinkless orientation is the problem of orienting edges of a graph such that each node has at least one outgoing edge. The problem can equivalently be described as edge grabbing, requiring that each node grabs one of its incident edges while no edge is grabbed more than once.
Hypergraph sinkless orientation is a natural generalization of that problem on hypergraphs, studied in~\cite{BMNSU_halls_thm_soda25}.

\begin{definition}[hypergraph sinkless orientation (HSO)]
\label{def:hso}
    Consider a hypergraph $H = (V,E)$. Let us define an orientation of a hyperedge $e \in E$ so that $e$ is outgoing for exactly one of its endpoints, and incoming for all of its other endpoints. An \emph{HSO} is an orientation of the hyperedges s.t.\ each node has at least one outgoing incident hyperedge.
    The \emph{HSO problem} is the problem of computing an HSO.
\end{definition}

Its state-of-the-art complexities in deterministic and randomized \LOCAL are as follows.

\begin{restatable}[{\cite{BMNSU_halls_thm_soda25}}]{theorem}{ThmHSO}
    \label{thm:HSO}
    There is a deterministic $O(\log_{\delta / r} n)$-round algorithm for computing an HSO in any $n$-node multihypergraph\footnote{A multihypergraph is simply a hypergraph in which the same hyperedge can appear more than once.} with minimum degree $\delta$ and maximum rank $r<\delta$. 
\end{restatable}

\begin{restatable}[{\cite{BMNSU_halls_thm_soda25}}]{theorem}{ThmHSOrand}
    \label{thm:HSO_randomized}
    There is a randomized $O(\log_{\delta / r} \log n + \log_{\delta / r} \delta)$-round algorithm for computing an HSO in any $n$-node multihypergraph with maximum rank $r$ and minimum degree $\delta>320 r \log r$ w.h.p. 
\end{restatable}

In other notation, $\Thso(n,\delta/r) \in O(\log_{\delta / r} n)$ and $\Thso(n,\delta/r,\delta) \in O(\log_{\delta / r} \log n + \log_{\delta / r} \delta)$.

\paragraph{MIS and coloring} Recall that in the Maximal Independent Set problem (MIS), the goal is to compute a set of nodes $I \subseteq V$ that is both 1) independent (for any two nodes $u,v \in I$, $uv \not \in E$), and 2) maximal (for any node $v \in V \setminus I$, $I \cup \set{v}$ is not independent, i.e., $\exists u \in N(v) \cap I$).

\begin{restatable}{theorem}{ThmMIS}
    \label{thm:local-mis}

    MIS admits deterministic LOCAL algorithms of complexity
    $\Otilde(\log^{5/3} n)$~\cite{GG_focs24},
    $O(\log n \log^2 \Delta)$~\cite{faour2023local}, and
    $O(\Delta + \log^* n)$~\cite{barenboim14distributed}.

    MIS admits randomized LOCAL algorithms of complexity
    $O(\log \Delta) + \Otilde(\log^{2} n)$~\cite{ghaffari16improved,GG_focs24},
    $\Otilde(\log^{5/3} \Delta) + \Otilde(\log^{5/3} \log n)$~\cite{ghaffari16improved,GG_focs24}, and
    $O(\log^3 \Delta + \log^2 \Delta \cdot \log \log n)$~\cite{ghaffari16improved,faour2023local}.
\end{restatable}

In the degree+1-list-coloring problem (D1LC), each node $v \in V$ has a list $\psi(v)$ of colors of size $\card{\psi(v)} = \deg(v) +1$. The goal is to assign each node $v \in V$ a color from its list such that adjacent nodes receive distinct colors. The problem is a stronger variant of the $(\Delta+1)$-coloring problem, where each node receives the same list of colors $\set{1,\dots,\Delta+1}$.
D1LC (a fortiori, $(\Delta+1)$-coloring) admits a simple reduction to MIS.

\begin{restatable}[{\cite[Section 6.1]{Luby1986}}]{theorem}{ThmColoringMISReduction} 
\label{thm:coloring-to-mis-reduction}
    Solving degree+1-list-coloring on a graph $G=(V,E)$ with $n$ nodes and maximum degree $\Delta$ reduces to solving MIS on a related graph $G'=(V',E')$ with $n' \leq n (\Delta+1)$ nodes and maximum degree $\Delta' \leq 2\Delta$, such that simulating one round of communication on $G'$ only takes $O(1)$ rounds of communication on $G$. In equations:
$\Tdlc(n,\Delta) \in O(\Tmis(n(\Delta+1),2\Delta))$ and $\Tdlcrand(n,\Delta) \in O(\Tmisrand(n(\Delta+1),2\Delta))$.
\end{restatable}

\begin{restatable}{theorem}{ThmColoring}
    \label{thm:local-d1lc}

    Degree+1-list-coloring admits deterministic LOCAL algorithms of complexity $\Otilde(\log^{5/3} n)$~\cite{GG_focs24}, $O(\log^2 \Delta \cdot \log n)$~\cite{ghaffari2021deterministic}, and $O(\sqrt{\Delta \log \Delta} + \log^* n)$~\cite{MausTonoyan_dc22}. 

    Degree+1-list-coloring admits randomized LOCAL algorithms of complexity $\Otilde(\log^{5/3} \log n)$~\cite{HalldorssonKNT22near,GG_focs24} and $O(\log^2 \Delta \cdot \log \log n)$~\cite{HalldorssonKNT22near,ghaffari2021deterministic}.
\end{restatable}

A \emph{layered graph} with $h$ layers is a graph whose nodes are partitioned into $h$ sets $V_1,\dots,V_h$. In a layered graph, the \emph{up-degree} $\updeg(v)$ of a node $v \in V_i, i \in [h]$ is the number of neighbors of $v$ in equal or lower layers, $\updeg(v) = \card{N(v) \cap \bigsqcup_{j\leq i} V_j}$. We denote by $\hDelta = \max_{v \in V} \updeg(v)$ the maximum updegree of the graph, and $\Tlayer(n,\hDelta,h)$ the deterministic complexity of list-coloring a layered graph with $h$ layers, maximum up-degree $\hDelta$, and whose nodes each have a list of colors $\psi(v)$ of size $\card{\psi(v)} = \updeg(v)+1$. Note that $\Tlayer(n,\hDelta,h) \in O(h \cdot \Tdlc(n,\hDelta))$. We similarly define $\Tlayerrand(n,\hDelta,h)$ for the randomized complexity of the problem, which is itself bounded by $O(h \cdot \Tdlcrand(n,\hDelta))$.

\paragraph{Degree-choosability and DCCs}
Our $\Delta$-coloring algorithm will make use of a concept called \emph{degree-choosable subgraphs (DCCs)}. A subgraph is degree-choosable if, for any assignment of color lists to its nodes, where each node is assigned a list of size equal to its degree, there exists a proper coloring in which each node receives a color from its list. As such, these subgraphs admit $\Delta$-colorings. 

\begin{definition}[Degree-choosability, DCCs~\cite{vizing1976choosability,erdosrubintaylor1979choosability}]
    \label{def:choosability}
    A graph $G=(V,E)$ is \emph{degree-choosable} if the following holds: for any assignment of lists of colors $(\psi(v))_{v \in V}$ to the nodes such that $|\psi(v)| \geq \deg(v)$ for each node $v \in V$, there exists a proper coloring of the nodes of $G$ such that each node $v$ receives a color from its list $\psi(v)$. An induced subgraph $H = G[U]$ of $G$ where $U \subseteq V$ is a \emph{degree-choosable component} (DCC) iff $H$ is connected and degree-choosable. 
\end{definition}

\begin{lemma}[\cite{vizing1976choosability,erdosrubintaylor1979choosability}]
    \label{lem:dcc-unless-clique-or-odd}
    Every $2$-connected graph is degree-choosable unless it is a clique or an odd cycle.
\end{lemma}

An absence of DCCs within some distance of a node constrains the topology around it. Notably, $\Delta$-regular graphs without small DCCs \emph{expand}, in the sense that each node in the graph has at least $\approx \Delta^{k/2}$ nodes at distance $k$ from itself, if $k$ is smaller than the size of the smallest DCC in the graph.

\begin{lemma}[{\cite[Lemma 1]{GHKM_dc21}}]
    \label{lem:uniqueBFS}
    Let $r$ be an integer, $G = (V,E)$ be a graph, and $v \in V$ be a node of $G$. Assume that the distance-$r$ neighborhood of $v$ contains no DCC. Then a depth-$r$ BFS rooted at $v$ is unique, i.e., for each node $u \in N^r(v)$, there exists a unique node $\parent(u)$ such that $\dist(\parent(u),v) = \dist(u,v) - 1$.
\end{lemma}

\begin{lemma}[{\cite[Lemma 8]{GHKM_dc21}}]
    \label{lem:DCCorexpand}
    For every node $v$ in a graph $G$, there exists either a node of degree less than $\Delta$ in distance $k$ of $v$ or a degree-choosable component of diameter at most $k$ in distance $k$ of $v$, or there are at least $(\Delta-1)^{\floor{k/2}}$ nodes at distance $k$ of $v$.
\end{lemma}

\section{Algorithm for General Graphs}
\label{sec:main-alg}
\subsection{Technical Roadmap}
\label{sec:roadmap}

In this section, we present in more details our approach for obtaining faster deterministic $\Delta$-coloring algorithms in the LOCAL model. The overall goal of this section is to prove \cref{thm:local-main}.

\ThmLOCALmain*

In the same way that only very specific graphs are not $\Delta$-colorable, the difficulty in computing a $\Delta$-coloring comes from only very specific configurations.
Prior work on the problem has highlighted that a node without any ``helpful subgraph'' (a DCC) or node of degree $< \Delta$ within some distance $k$ has exponentially many neighbors at distance $k$ (\cref{lem:DCCorexpand}).
At a high level, we aim to generate some flexibility similar to that provided by a nearby DCC throughout the graph, even in parts without small DCCs.
We do so by coloring a few nodes (set $I$ in \cref{lem:finish-coloring}) throughout the graph with color $1$, leaving any node with two $1$-colored neighbors with only $\Delta-2$ uncolored neighbors but $\Delta-1$ colors to choose from (set $T_I$ in \cref{lem:finish-coloring}).

Having ensured that all nodes in the graph have a short path to either a node with two $1$-colored neighbors, a DCC, or a node of degree $<\Delta$, obtaining a $\Delta$-coloring is quite straightforward (\cref{lem:finish-coloring})

\begin{restatable}{lemma}{LemFinishColoring}
\label{lem:finish-coloring}
  Let $G=(V,E)$ be a graph, $k$ and $d$ be integers. Suppose that we have computed $\fS \subseteq 2^V$ and $I \subseteq V$ such that:
  \begin{enumerate}
    \item for each member $U\in \fS$, the induced subgraph $G[U]$ is either a DCC of weak diameter at most $k$, or a single node of degree strictly less than $\Delta$ ($\fS$ is $\Delta$-completable),
    \item for any two members $U,U' \in S$, $U\cap (U' \sqcup N(U')) = \emptyset$ (independence of $\fS$),
    \item $I$ is an independent set, and $I \cap \bigsqcup_{U \in S} U = \emptyset$ ($I$ independent and outside $\fS$),

    Let $G'=G[V\setminus I]$, $T_I = \set{ v : \card{N(v) \cap I} \geq 2}$ be the set of nodes with two neighbors in $I$, $W = T_I \cup \bigsqcup_{U \in S} U \subseteq V\setminus I$ be the set of nodes in either $\fS$ or $T_I$.

    \item for any node $v \in V \setminus I$, $\dist_{G'}(v,W)\leq d$ (proximity of $\fS$ and $T_I$).
\end{enumerate}
  Then a $\Delta$-coloring of $G$ can be computed in $\Tlayer(n,\Delta-1,d)+O(k)$ rounds of \LOCAL.
\end{restatable}

The approach of \cref{lem:finish-coloring} is rather generic and is the backbone of several prior works on $\Delta$-coloring.
While not presented as such in their respective papers, the deterministic algorithms of \cite{GHKM_dc21} and \cite{ghaffari2021deterministic} can be reinterpreted as computing $\fS$ and $I$ such that $I = \emptyset$ and $k = O(\log_\Delta n)$ in both cases, and $d = O(\log^2_\Delta n)$ and $d = O(\log^2 n / \log \Delta)$, respectively. 
The approach of \cite{bourreau2025faster} diverges from \cref{lem:finish-coloring} in that its
deterministic algorithm for $\Delta$-coloring computes sets $\fS$ and $I$ such that $I=\emptyset$, $k = O(\log_\Delta n)$, $d = O(\log \Delta + \log_\Delta n \log^* n)$, but the subgraphs induced by $\fS$ are no longer necessarily independent. Instead, $\fS$ is chosen so that any partial $\Delta$-coloring not involving the nodes $\bigcup_{U \in S} U$ contained in $\fS$ can be extended to $\fS$ in $\Tlayer(n,\Delta-1,k)+O(k)$ rounds of \LOCAL.
Note that all the just mentioned deterministic algorithms do not make use of the set $I$.
A recent paper focusing on dense graphs uses multiple sets $I$~\cite{JM_podcba25}, and builds $\fS$ and $I$ with distance $d=O(1)$ and DCCs of size $k=O(1)$.
In contrast, randomized algorithms sample nodes throughout the graph to populate $I$, achieve a smaller parameter $k=O(\log_\Delta \log n)$, and smaller distances $d=O(\log^2 \log n / \log \Delta)$~\cite{GHKM_dc21} or $d=O(\log_\Delta \log n \log^* n)$~\cite{bourreau2025faster}.
Several works also make use of the \emph{shattering} technique to reduce the problem to smaller instances of size $\poly(\log n)$.
An algorithm from one of these works~\cite{fischer2023fast} computes multiple independent sets $I$, and essentially builds $\fS$ and $I$s with distance $d=O(1)$ and size of DCCs $k=O(1)$ when $\Delta$ is a large enough $\poly \log n$, and achieves a complexity of $O(\log^* n)$.

The proof of \cref{lem:finish-coloring} is not novel and rather simple. We present it later in the paper, in \cref{sec:finish-coloring}.

With \cref{lem:finish-coloring} as final step in our algorithm, computing a $\Delta$-coloring is only a matter of computing sets $\fS$ and $I$ as in its statement, with parameters $d$ and $k$ as small as possible.
Herein lies the main contribution of this paper. We show how to compute sets $\fS$ and $I$ with distance $d=O(1)$ and size of DCCs $k=O(1)$ in a runtime linear in that of MIS and HSO.

\begin{restatable}{lemma}{LemComputingDCCandTnodes}
\label{lem:computing-DCC-T-nodes}
    There is a \LOCAL algorithm for computing sets $\fS$ and $I$ as in \cref{lem:finish-coloring} with distance $d=O(1)$ and size of DCCs $k=O(1)$ in $O(\Tmis(n,\poly(\Delta)) + \Thso(n,\poly(\Delta)))$ rounds.
\end{restatable}

Computing $\fS$ is relatively straightforward, finding $I$ is where the effort is. We only look for nodes to add to $I$ in parts of the graph without small DCCs and nodes of degree $< \Delta$. A property of these areas is that they exhibit local expansion, i.e., each node $v$ in them has a minimum of $\approx (\Delta-1)^{t/2}$ nodes at distance $t$ from itself, for $t$ below some constant. 
We partition these parts of the graph into small diameter clusters, and define a virtual multihypergraph from them.
This hypergraph is such that a hypergraph sinkless orientation on this hypergraph immediately yields sets $\fS$ and $I$ as in \cref{lem:finish-coloring,lem:computing-DCC-T-nodes}.
All that remains is computing a cluster partition with the right properties, so that the resulting hypergraph admits a hypergraph sinkless orientation, and that such an orientation can be computed efficiently.

\begin{restatable}{definition}{DefClusterPartition}
    \label{def:cluster-partition}
    Let $\gdcc$, $\gflex$, and $\glink$ be integers, and let the nodes of a graph $G=(V,E)$ be partitioned into three sets $\Vdcc$, $\Vflex$, and $\Vlink$ such that
    \begin{enumerate}
        \item $\Vdcc$ (resp.\ $\Vflex$, $\Vlink$) is partitioned into a set $\Cdcc$ (resp.\ $\Cflex$, $\Clink$) of disjoint clusters,
        \item each cluster $C \in \Cdcc$ has an induced subgraph $S_C$ that is either a DCC or a node of degree strictly less than $\Delta$, and is independent of the $S_{C'}$ of any other cluster $C'\in \Cdcc \setminus \set{C}$,
        \item each cluster $C \in \Cflex$ has a designated node $z_C$ called its \emph{center}.
        \item for each cluster $C \in \Cflex$ and node $v \in C$, $v$ has a path $P$ to a node $w\in N(C) \cap (\Vdcc \cup \Vlink)$ such that for each node $u \in P \cap C$ on that path, $\dist_C(u,z_C) > \dist_C(v,z_C)$,
        \item each cluster $C \in \Clink$ is adjacent to a non-link cluster ($N(C) \cap (\Vdcc \cup \Vflex) \neq \emptyset$),
        \item each cluster $C \in \Cdcc$ (resp.\ $C \in \Cflex$, $C \in \Clink$) is of strong diameter $\gdcc$ (resp.\ $\gflex$, $\glink$),
    \end{enumerate}
    We call $(\Vdcc,\Vflex,\Vlink)$ a $(\gdcc,\gflex,\glink)$-$\Delta$-coloring cluster partition.\end{restatable}
For brevity, we will refer to this construction simply as a ``cluster partition''.

The set $\fS$ is simply made of the internal flexible subgraphs $(S_C)_{C \in \Cdcc}$ contained in the DCC clusters $\Cdcc$. The independent set $I$ is created by the flex clusters, $\Cflex$. In each flex cluster, we identify a set of nodes that have the potential to select two unconnected neighbors that could join $I$, whose removal does not disconnect their common neighbor from the center of the cluster. We call such nodes \emph{gatherers}.

\begin{restatable}{definition}{DefGatherer}[Gatherers]
\label{def:gatherer}
    Let $C \in \Cflex$ be a cluster in a $\Delta$-coloring cluster partition, and consider $v \in V(C) \setminus z_C$. Let $\parent(v)$ be the last node on a shortest path from $z_C$ to $v$ in $C$, and $T_v = (N(v) \cap C) \setminus \set{\parent(v)}$. $v$ is a \emph{gatherer} for $C$ if there exists $w,w' \in T_v$ s.t.\ $ww' \not \in E$.
    
    Additionally, let $C' = C[V(C) \setminus T_v]$. Consider the set of nodes $D_v = \set{w \in V(C) : \dist_{C'}(w,z_C) = \infty}$ that are disconnected from the cluster's center $z_C$ when removing $T_v$ from the cluster $C$.
    We denote by $\conflict(v) \subseteq \Vlink$ the set of link clusters conflicting with $v$, defined by $\conflict(v) = \set{C \in \Clink : N(T_v \sqcup D_v) \cap C \neq \emptyset}$.
    Two gatherers $v,v' \in C$ are said to be independent iff $(T_v \sqcup D_v) \cap (T_{v'} \sqcup D_{v'}) = \emptyset$,  
\end{restatable}

Intuitively, each gatherer is a way that a flex cluster can take two of its nodes to add to $I$, so that coloring the two nodes allows to color most of the cluster (and adjacent link clusters) with only colors from $\set{1,\dots,\Delta}$.
However, each gatherer has its own drawbacks: a selected gatherer in a cluster $C \in \Cflex$ cuts some paths between the cluster's center node $z_C$ and some adjacent link clusters, and even between $z_C$ and some nodes from $C$. These disconnected paths imply that such nodes and adjacent link clusters must rely on a different short path than one going through the gatherer and the flex cluster's center to reach a node with two neighbors in $I$. Independence between gatherers is important to guarantee that when a gatherer is not selected, all the nodes and adjacent link clusters that would have been disconnected from selecting this gatherer actually have a short free path through that gatherer to the center of the cluster.

The competition between flex and link clusters for whether each gatherer selects two of its neighbors to add to $I$ is naturally interpretable as an instance of hypergraph sinkless orientation, where the nodes correspond to clusters, and each gatherer corresponds to a hyperedge.
The idea of using HSO as a subroutine to create some flexibility throughout the graph has appeared in two very recent works about simpler variants of the $\Delta$-coloring problem~\cite{JM_podcba25,JMS_arxiv25}.

\begin{restatable}{definition}{DefClusterHypergraph}
    \label{def:cluster-hypergraph}
    Consider a $\Delta$-coloring cluster partition $(\Vdcc,\Vflex,\Vlink)$ and let each $C \in \Cflex$ have a set of independent gatherers $\gath(C) \subseteq V(C)$.
    The bipartite cluster hypergraph of $(\Vdcc,\Vflex,\Vlink,(\gath(C)_{C \in \Cflex}))$ is the hypergraph $H=(\Cflex \times \Clink,E')$ containing the following hyperedges:
    \begin{enumerate}
        \item for each $C \in \Cflex \cup \Clink$ such that $C$ is touching a DCC-cluster ($N(C) \cap \Vdcc \neq \emptyset$), $H$ has a rank-$1$ hyperedge with just $C$ ($\set{C} \in E'$), 
        \item for each $C \in \Cflex$ and gatherer $v \in \gath(C)$ conflicting with $c = \card{\conflict(v)}$ link clusters, $H$ has a rank-$(c+1)$ hyperedge containing $C$ and $\conflict(v)$ ($(\set{C} \cup \conflict(v)) \in E'$), 
        \item for each $C \in \Clink$ and $C' \in \Cflex$ such that $C$ and $C'$ are adjacent ($N(C) \cap C' \neq \emptyset$) and no gatherer of $C'$ conflicts with $C$ ($\forall v \in \gath(C'), C \not \in \conflict(v)$), $H$ has a rank-$1$ hyperedge with just $C$ ($\set{C} \in E'$).
    \end{enumerate}    
\end{restatable}

Note that the cluster hypergraph is only ``bipartite'' in the sense that its nodes are partitioned into two known disjoint sets $\Cflex$ and $\Clink$, and not in the usual sense that no edge connects two nodes of the same set.

\begin{restatable}{lemma}{LemClusterHypergraphOrientation}
    \label{lem:cluster-hypergraph-hso}
    Consider a $(\gdcc,\gflex,\glink)$-$\Delta$-coloring cluster partition $(\Vdcc,\Vflex,\Vlink)$, sets of independent gatherers $(\gath(C))_{C \in \Cflex}$, and suppose we are given a hypergraph sinkless orientation of their bipartite cluster hypergraph $H$. Suppose that for each cluster $C \in \Vdcc$ its internal flexible subgraph $S_C$ is of weak diameter at most $k$.
    Then sets $\fS$ and $I$ with distance $d= O(\glink + \gflex +\gdcc)$ and DCC diameter $k$ as in \cref{lem:finish-coloring,lem:computing-DCC-T-nodes} can be computed in $O(\glink + \gflex + \gdcc)$.
\end{restatable}

For a hypergraph to even admit a sinkless orientation, and also so that it can be computed fast in the \LOCAL model, we need to be mindful of the minimum degree and rank of the hypergraph stemming from our cluster partition.
We show that we can ensure a minimum degree and a rank both of order $\poly(\Delta)$, with a $\poly(\Delta)$-ratio between the two.

\begin{restatable}{lemma}{LemComputingClusterPartition}
    \label{lem:cluster-partition-computing}
    Computing a $(\gdcc,\gflex,\glink)$-$\Delta$-coloring cluster partition $(\Vdcc,\Vflex,\Vlink)$ and gatherers $(\gath(C))_{C \in \Cflex}$ whose corresponding bipartite cluster multihypergraph $H=(\Cflex \times \Clink,E')$ and parameters $\gdcc$, $\gflex$, and $\glink$ satisfying 
    \begin{enumerate}
        \item $\gdcc = O(1)$, $\gflex=O(1)$, and $\glink = O(1)$,
        \item each hyperedge in $H$ has rank at most $1+(\Delta-1)^3$,\item each node in $H$ either has an incident hyperedge of degree $1$, or degree at least $(\Delta-1)^4$,
    \end{enumerate}    
    can be done in $O(\Tmis(n,\poly(\Delta)))$.
\end{restatable}

To compute this cluster partition, we first add to $\Vdcc$ all the nodes within some constant distance $\adcc$ from a node of degree less than $\Delta$ or a DCC of constant weak diameter $\leq \adcc$. This ensures that for any remaining node $v\in V \setminus \Vdcc$, any ball of constant radius $r \leq \adcc$ centered on $v$ either intersects $\Vdcc$, or contains $\Delta^{\Omega(r)}$ nodes.

By computing an MIS $\Sflex$ over the power graph $(G')^{\bflex}$ of the subgraph $G' = G[V \setminus \Vdcc]$ induced by the remaining nodes, and taking as flex clusters the balls of radius $\aflex < \bflex$ centered on the MIS, we obtain flex clusters that each either contain at least $\Delta^{\Omega(\aflex)}$ nodes, or touch a DCC cluster. For flex clusters that do not touch a DCC cluster, the $\poly(\Delta)$ number of nodes and absence of DCCs makes it possible to find $\poly(\Delta)$ independent gatherers in each of them. Selecting the gatherers close to the border of a flex cluster ensures that each gatherer can only cut a small $\poly(\Delta)$ number of paths between its flex cluster's center and adjacent link clusters.

Finally, nodes from $V \setminus (\Vdcc \sqcup \Vflex)$ are partitioned into link clusters. We do so by computing an MIS $\Slink$ over the power graph $(G'')^{\blink}$ where $G'' = G[V \setminus (\Vdcc \sqcup \Vflex)]$, and having one cluster per node in the MIS, where each node $v \in \Vlink = V \setminus (\Vdcc \sqcup \Vflex)$ joins the cluster of the MIS node whose BFS reaches $v$ first as a BFS is performed from each MIS node. A sufficient gap $\alink$ between $\bflex$ and $\aflex$ (leaving some distance between flex clusters), and a sufficiently high distance $\blink$ between nodes in $\Slink$, ensures that each link cluster is adjacent to a high enough $\poly(\Delta)$ number of flex clusters, or at least one DCC cluster.

Put together, computing the cluster partition and sets of independent gatherers, computing a hypergraph sinkless orientation on the corresponding cluster hypergraph, constructing sets $\fS$ and $I$ from that, and finishing the coloring by coloring nodes in decreasing order of distance to $\fS$ and $I$ before coloring $\fS$, results in the following algorithm for $\Delta$-coloring (\cref{alg:main-algorithm}).

\begin{algorithm}[htb!]
\caption{Algorithm for $\Delta$-coloring}
\label{alg:main-algorithm}
\begin{algorithmic}[1]
\State \BuildClusterPartition \Comment{(\cref{sec:build-clusters})}
\Statex Consider the virtual cluster hypergraph
\State \VirtualHSO \Comment{(\cref{sec:hso-on-cluster-hypergraph})}
\State \LayerColoring \Comment{(\cref{sec:finish-coloring})}
\end{algorithmic}
\end{algorithm}

The construction of the cluster partition is explained in \cref{sec:build-clusters}. This part of the algorithm is where the most effort is.
We explain the cluster hypergraph and its parameters as stemming from the cluster partition in \cref{sec:hypergraph-properties}. In \cref{sec:hso-on-cluster-hypergraph}, we explain how computing a hypergraph sinkless orientation on this hypergraph allows us to compute sets $\fS$ and $I$ as we need for our $\Delta$-coloring algorithm. Finally, we explain in \cref{sec:finish-coloring} how to make use of the computed sets $\fS$ and $I$ and obtain our full algorithm for $\Delta$-coloring.

\subsection{Constructing the Cluster Graph}
\label{sec:build-clusters}

In this subsection, we show how to construct the clusters $\Cdcc, \Cflex, \Clink$, the way they are defined in \cref{def:cluster-partition}. 
\begin{algorithm}[htb!]
\caption{\BuildClusterPartition}
\label{alg:build-cluster-partition}
\begin{algorithmic}[1]
\State \BuildCdcc \Comment{(\cref{sec:dcc-clusters})}
\State \BuildCflex \Comment{(\cref{sec:flex-clusters})}
\State \BuildClink  \Comment{(\cref{sec:link-clusters})}
\State \FindGatherers \Comment{(\cref{sec:gatherers})}
\end{algorithmic}
\end{algorithm}

\subsubsection{Construction of DCC Clusters} \label{sec:dcc-clusters}

The first set of clusters built by the algorithm is $\Cdcc$. Each of its clusters $C \in \Cdcc$ is guaranteed to contain a flexible subgraph $S_C$, that is either a small diameter DCC or a node of degree less than $\Delta$. Conversely, all small diameter DCCs and nodes of degree less than $\Delta$ are contained in the union of all DCC clusters.

For the construction of this first set of clusters, each node gathers the complete information about its distance-$\adcc$ neighborhood. Any node of degree $\Delta$ that is part of at least one DCC of diameter less than $\adcc$ containing only degree-$\Delta$ nodes selects one such DCC and broadcasts this choice to its entire distance-$\adcc$ neighborhood.
We break symmetry between the selected DCCs, together with the nodes of degree less than $\Delta$, so as to obtain a set $\Wdcc$ of independent DCCs and nodes of degree less than $\Delta$. 
This corresponds to computing an MIS on a virtual graph $\Gdcc$, whose vertices are the selected DCCs and nodes of degree less than $\Delta$, and edges represent adjacency.
We then grow clusters from this set $\Wdcc$, so that the grown clusters cover all nodes of degree less than $\Delta$ and small diameter DCCs.

\begin{algorithm}[htb!]
\caption{\BuildCdcc}
\label{alg:build-Cdcc}
\begin{algorithmic}[1]
\Statex \textbf{Input:} Positive integer $\adcc$\State Each node of degree $\Delta$ learns its distance-$\adcc$ neighborhood.
\State If part of one or more DCCs of weak diameter at most $\adcc$ and containing only nodes of degree $\Delta$, a node of degree $\Delta$ selects a single one arbitrarily.\State Let $\Sdcc$ be the set of selected DCCs and of nodes of degree at most $\Delta - 1$.
\State Let $\Gdcc$ be the virtual graph with nodes $\Sdcc$, and edges between touching subgraphs.
\State Compute an MIS over $\Gdcc$, let $\Wdcc \subseteq \Sdcc$ be its members.
\State Initialize the set of clusters $\Cdcc \gets \Wdcc$. For each cluster $C \in \Wdcc$, its flexible subgraph $S_C$ is $C$ itself at this stage.
\State Grow the clusters by performing a $(\adcc+1)$-depth BFS from each cluster $C \in \Cdcc$, breaking symmetry between the different BFS by depth, then ID.
\end{algorithmic}
\end{algorithm}

\begin{restatable}{lemma}{LemBuildCdcc}
\label{lem:build-Cdcc-result}
    \Cref{alg:build-Cdcc} computes a set of clusters $\Cdcc$ in $O(\adcc\cdot\Tmis(n,\Delta^{2\adcc+2})$ rounds of \LOCAL such that
    \begin{enumerate}
        \item Each cluster $C\in \Cdcc$ contains a subgraph $S_C$ that is either a node of degree $< \Delta$ or a DCC,
        \item For any two distinct clusters $C,C' \in \Cdcc$, $N(S_C) \cap S_{C'} = \emptyset$,
        \item For any integer $r\geq 0$ and node $v$ s.t.\ $\dist(v,\Vdcc) > r$ where $\Vdcc = \bigsqcup_{C \in \Cdcc} V(C)$, the distance-$r$ neighborhood of $v$ contains no node of degree $<\Delta$ or DCC of weak diameter $\adcc$ or less.
    \end{enumerate}
\end{restatable}
\begin{proof}
    Let us bound the degree of the virtual graph $\Gdcc$. Consider a node of this graph, i.e., a member $S\in \Sdcc$.
    $S$ is either a single node $v \in V$ of degree $< \Delta$ or a DCC of weak diameter $\leq \adcc$ that was selected by one or more of its constituent nodes. The bound on the weak diameter implies that $S$ contains at most $1+\Delta^{\adcc}$ nodes.
    Every other member $S'\in \Sdcc$ connected to $S$ in $\Gdcc$ intersects with $S$ or its neighborhood, i.e., $S' \cap (S \cap N(S)) \neq \emptyset$. There are no more than $\Delta+\Delta^{\adcc+1}$ neighbors of $S$, and as a result, at most this many nodes of degree $\Delta-1$ connected to $S$ in $\Gdcc$. For each DCC $S' \in \Sdcc$ intersecting $S$ or its neighborhood, this DCC $S'$ was selected by one of its nodes. As all DCCs have weak diameter at most $\adcc$, any node that selected a DCC $S'$ intersecting $S$ or its neighborhood is at distance at most $2\adcc +1$ from any node in $S$. As there are at most $1+\Delta^{2\adcc+1}$ nodes in the distance-$2\adcc+1$ neighborhood of any node in $G$, and adding the different bounds together, we get that $\Gdcc$ has a maximum degree of $\Delta^{2\adcc+2}$. As a round of communication on $\Gdcc$ can be simulated in $O(\adcc)$ rounds of communication on $G$, computing an MIS of $\Gdcc$ only takes $O(\adcc \cdot \Tmis(n,\Delta^{2\adcc+2}))$ rounds of communication on $G$. Growing the clusters with a $(\adcc+1)$-depth BFS only takes $O(\adcc)$ additional rounds, giving the claimed round complexity. The first two properties in the list follow from the fact that each cluster is grown from a member $S \in \Wdcc$, where $\Wdcc \subset \Sdcc$ in an independent set of $\Gdcc$.

    Any node $v$ of degree $<\Delta$ that is not part of $\Wdcc$ is necessarily at distance $0$ or $1$ from an element of $\Wdcc$, as $\set{v} \in \Wdcc$. Any DCC $S$ of weak diameter $\adcc$ or less is at distance $\adcc+1$ or less from a DCC in $\Wdcc$: if $S \in \Sdcc$, $\dist(S,\Wdcc)\leq 1$ follows from the fact that $\Wdcc$ is an MIS of $\Gdcc$ of vertex set $\Sdcc$ with edges between touching DCCs; if $S\not \in \Sdcc$, then each of its constituent nodes chose a different DCC of weak diameter $\leq \adcc$ that they are part of to join $\Sdcc$, so $\dist(S,\Wdcc) \leq \dist(S,\Sdcc) + \adcc + \dist(\Sdcc,\Wdcc) \leq \adcc+1$. As a result, every node of degree $< \Delta$ and DCC of weak diameter $\leq \adcc$ is contained in the distance-$\adcc+1$ neighborhood of $\Wdcc$, i.e., in $\Vdcc = \bigsqcup_{C \in \Cdcc} V(C)$. A node at distance $>r$ from $\Vdcc$ has a distance-$r$ neighborhood that does not intersect $\Vdcc$, and as such, is free of degree $< \Delta$ nodes and DCCs of weak diameter $< \adcc$.
\end{proof}

In particular, the absence of small DCCs around each node in $V \setminus \Vdcc$ implies some local expansion around each node by \cref{lem:DCCorexpand}.

\begin{restatable}{claim}{LemNoDccExpand}
    \label{claim:no-dcc-expand}
    Let $r < \adcc/2$ be an integer and $v \in V \setminus \Vdcc$ be a node outside of the $\Cdcc$ clusters. Consider the distance-$r$ ball around $v$, $B(v,r) = \set{u \in V : \dist(u,v) \leq r}$, and the distance-$r$ sphere centered at $v$, $S(v,r) = \set{u \in V : \dist(u,v) = r}$. Either $B(v,r) \cap \Vdcc \neq \emptyset$, or $\card{S(v,r)} \geq (\Delta-1)^{\floor{r/2}}$.
\end{restatable}
\begin{proof}
    Follows immediately from \cref{lem:build-Cdcc-result,lem:DCCorexpand}: if $B(v,r) \cap \Vdcc = \emptyset$, the distance-$r$ neighborhood of $v$ is free of DCCs of weak diameter $\adcc$ or less as well as of nodes of degree less than $\Delta$, which gives the expansion.
\end{proof}

\subsubsection{Construction of Flex Clusters} \label{sec:flex-clusters}

The second set of clusters built by the algorithm is $\Cflex$. For this set of clusters, we use the fact that nodes close to a small DCC or a node of degree less than $\Delta$ were all added to $\Cdcc$, and the resulting \cref{claim:no-dcc-expand} that any node outside $\Vdcc$ is either close to $\Vdcc$, or has some local expansion, i.e., it has exponential-in-$k$ many nodes at distance $k$ for $k$ below some constant.

\begin{algorithm}[htb!]
\caption{\BuildCflex}
\label{alg:build-Cflex}
\begin{algorithmic}[1]
\Statex \textbf{Input:} Positive integers $\aflex$, $\bflex$, and $\alink$ such that $\bflex \geq 2(\aflex+\alink+1)$.\State Let $G' = G[V \setminus \Vdcc]$, compute an MIS $\Sflex$ on $(G')^{\bflex}$.
\State For each $v \in \Sflex$, let $C_v = \set{u \in V \setminus \Vdcc: \dist_{G'}(u,v) \leq \aflex}$.
\State Set $\Cflex = \set{C_v : v \in \Sflex}$, and for each $C_v \in \Cflex$, let its center $z_{C_v}$ be $v \in \Sflex$.
\end{algorithmic}
\end{algorithm}

\begin{claim}
\label{claim:expansion-Cflex}
    For each cluster $C \in \Cflex$, either $N(C) \cap \Vdcc \neq \emptyset$, or $C$ contains at least $(\Delta-1)^{\floor{k/2}}$ nodes at distance $k$ from its center node $z_C$ for every $k \leq \aflex$.
\end{claim}
\begin{proof}
    Consider a cluster $C \in \Cflex$ and its center $z_C$. If $\dist(z_C,\Vdcc) \leq \aflex$, $C$ touches $\Vdcc$. Otherwise, \cref{claim:no-dcc-expand} gives the bound on the number of nodes at distance $k$ from the center.
\end{proof}

\subsubsection{Construction of Link Clusters} \label{sec:link-clusters}

Finally, we build the third and last set of clusters $\Clink$. All the nodes that are not yet part of a cluster (nodes from $V \setminus (\Vdcc \sqcup \Vflex)$) join a cluster from this last set.

\begin{algorithm}[htb!]
\caption{Algorithm \BuildClink}
\label{alg:build-Clink}
\begin{algorithmic}[1]
\Statex \textbf{Input:} Positive integer $\blink$ such that $\blink \geq \alink + (3/2)\bflex$\State Let $G'' = G[V \setminus (\Vdcc \sqcup \Vflex)]$, compute an MIS $\Slink$ on $(G'')^{\blink}$.
\State Perform a depth-$\blink$ BFS in $G''$ from each $v \in \Slink$.
\State For each $v \in \Slink$, let $C_v$ be the set of nodes reached by $v$'s BFS first, breaking ties by IDs after depth.
\State Set $\Clink = \set{C_v : v \in \Slink}$.
\end{algorithmic}
\end{algorithm}

\begin{claim}
\label{claim:expansion-Clink}
    For each cluster $C \in \Clink$, either $N(C) \cap \Vdcc \neq \emptyset$, or there are at least $\frac 1 2 (\Delta-1)^{\floor{\alink/2}}$ edges between $C$ and $\Cflex$.
\end{claim}
\begin{proof}
    Consider a node $v \in \Slink \subseteq V \setminus (\Vdcc \sqcup \Vflex)$, of corresponding link cluster $C_v \in \Clink$, and consider $\dist_{G'}(v,\Sflex)$ its distance to the nearest center of a $\Cflex$ cluster in $G' = G[V \setminus \Vdcc]$. 
    $\Sflex$ is an MIS of $(G')^{\bflex}$, therefore, there exists $z_C \in \Sflex$ such that $\dist_{G'}(v,z_C) \leq \bflex$, and for all $z_{C'} \in \Sflex \setminus \set{z_C}$, $\dist_{G'}(v,z_{C'}) \geq \bflex + 1 - \dist_{G'}(v,z_C)$.
    We show that either $v$ has a path of length at most $\floor{\bflex/2}$ in $G'$ to a node $v'$ such that either $v' \in N(\Vdcc)$ (making $v$'s cluster $C_v$ adjacent to a DCC cluster) or $\min_{C \in \Cflex} \dist_{G'}(v',z_C) \geq \floor{\bflex/2} > \aflex + \alink$ (making the distance-$\alink$ ball around $v'$ in $G'$ fully in $V \setminus (\Vdcc \sqcup \Vflex)$).

    Take any node $w \in V \setminus (\Vdcc \sqcup \Vflex)$ such that $\min_{C \in \Cflex} \dist_{G'}(w,z_C) < \floor{\bflex/2}$. Let $\Cstar \in \Cflex$ be the cluster such that $\dist_{G'}(w,z_{\Cstar}) = \min_{C \in \Cflex} \dist_{G'}(w,z_C)$.
    As $w$ is not in $\Vdcc$, it has degree $\Delta \geq 3$ in the original graph $G$. Assume that $w$ has no neighbor in $\Vdcc$, i.e., $N(w) \subseteq V \setminus \Vdcc$. Suppose that $w$ has no neighbor $w' \in N(w)$ such that $\min_{C \in \Cflex} \dist_{G'}(w',z_C) > \min_{C \in \Cflex} \dist_{G'}(w,z_C)$. Then for each $w' \in N(w)$, there exists a cluster $C' \in \Cflex$ such that $\dist_{G'}(w',z_{C'}) \leq \dist_{G'}(w,z_{\Cstar}) < \floor{\bflex / 2}$. Any such cluster $C'$ must be $\Cstar$, from the minimum distance of $\bflex +1$ between any two distinct $z_C, z_{C'} \in \Sflex$. Hence, $w$ has $\Delta \geq 3$ paths in $G'$ of length at most $\bflex/2$ with distinct first nodes to the same node $z_C$, implying the existence of a DCC of diameter $\leq \bflex$ in $G'$, which is impossible as all DCCs of weak diameter up to $\adcc \geq \bflex$ are in $\Vdcc$. Hence, $w$ either has a neighbor in $\Vdcc$, or it has a neighbor that is one hop further away from $\Sflex$ in $G'$. Note that in this second case, the neighbor of $w$ is in $V \setminus (\Vdcc \sqcup \Vflex)$ just like $w$, i.e., an edge connecting the two nodes is also present in $G'' = G[V \setminus (\Vdcc \sqcup \Vflex)]$.

    By repeatedly applying this observation, any node $v \in \Slink$ either has a path in $G''$ to a node $w_v$ at distance $\geq \floor{\bflex / 2}$ from $\Sflex$ in $G'$, or to a node $w_v$ adjacent to $\Vdcc$ in $G$. In the second case ($w_v \in N(\Vdcc$)), $v$'s cluster $C_v \in \Clink$ is adjacent to a DCC cluster. In the first case, note that the distance-$\alink$ neighborhood of $w_v$ in $G'$ does not intersect $\Vflex$. By \cref{claim:no-dcc-expand}, either $w_v$ has a path of length $\leq \alink$ to a node in $\Vdcc$, or it has $(\Delta-1)^{\floor{\alink/2}}$ nodes at distance $\alink$. In the first case, $v$'s cluster is adjacent to a DCC cluster. In the second case, note that the distance-$\alink$ sphere centered at $w_v$ is fully contained in $v$'s cluster, since $\blink \geq \alink + \floor{\bflex/2}$.

    Consider a BFS tree rooted at $w_v$ of depth $\alink + \bflex$ in $G'$, and for each node $u \in S(w_v,\alink)$ (i.e., at distance $\alink$ from $w_v$), consider the subtree rooted at $u$ in the BFS tree. Note that the BFS tree is unique, as guaranteed by \cref{lem:uniqueBFS}. Let $U_u$ be the nodes of that tree except $u$, i.e., $U_u$ contains the nodes in $V \setminus \Vdcc$ such that a shortest path between them and $w_v$ in $G'$ goes through $u \in S(w_v,\alink)$ and is of length $\in [\alink+1, \alink + \bflex]$.

    We show that the graphs induced by the sets $(U_u)_{u \in S(w_v,\alink)}$ are mostly disjoint connected components. Consider a node $u \in S(w_v,\alink)$. We show that there can be at most one edge between $U_u$ and $\bigsqcup_{u' \in S(w_v,\alink)\setminus \set{u}} U_{u'}$. Indeed, suppose there are two. Call $p, p' \in U_u$ and $q,q'\in \bigsqcup_{u' \in S(w_v,\alink)\setminus \set{u}} U_{u'}$ the nodes such that $pq,p'q' \in E$. Consider the common ancestor $r$ of highest depth to $q$,$q'$ and $u$. Taking the union of nodes on shortest paths between $r$ and $q$, $r$ and $q'$, $r$ and $p$, and $r$ and $p'$, gives a $2$-connected induced subgraph that is neither a clique, nor a cycle, and thus a DCC (\cref{lem:dcc-unless-clique-or-odd}) of weak diameter at most $2(\alink + \bflex) \leq \adcc$. As all such DCCs are contained in $\Vdcc$, each set $U_u$ has at most one edge to another set $U_{u'}$. Thus the subgraph induced by $\bigsqcup_{u \in S(w_v,\alink)} U_{u}$ in $G'$ has at least $\card{S(w_v,\alink)}/2$ connected components.

    Finally, we show that each of the $\card{S(w_v,\alink)}/2$ connected components induced by $\bigsqcup_{u \in S(w_v,\alink)} U_{u}$ in $G'$ contains at least one $z_C$, $C \in \Cflex$. Indeed, if not, consider one such connected component $G'[U_u]$ or $G'[U_u \sqcup U_{u'}]$ and a node $t$ in it at distance $\bflex + \alink$ from $w_v$. $t$ is a node from $G'$, and since $\Sflex$ is an MIS of $(G')^{\bflex}$, there must be a node $z_C \in \Sflex$ at distance at most $\bflex$ from $t$. As the entire distance-$\bflex$ neighborhood of $t$ is contained in its connected component $G'[U_u]$ or $G'[U_u \sqcup U_{u'}]$, there exists some $z_C$ in that connected component. Therefore, for each connected component, there is a path between $w_v$ and a node $t' \in V \setminus \Vdcc$ in $G''$ (i.e., taking only nodes from $V \setminus (\Vdcc \sqcup \Vflex)$) where $t' \in N(\Vflex)$, and thus, there are at least $\card{S(w_v,\alink)}/2$ edges between $C_v \in \Clink$ and $\Vflex$.
\end{proof}

\subsubsection{Finding Independent Gatherers}
\label{sec:gatherers}

We show in this section how to construct the gatherers from \cref{def:gatherer}. The gatherers are nodes from $\Cflex$ that display some structural properties.
To measure how many gatherers can be contained in each cluster of $\Cflex$, we will use the local expansion property described in \cref{claim:expansion-Cflex} and the fact that every node in this cluster has degree $\Delta$.

\begin{figure}[ht]
    \centering
\includegraphics[width=0.3\linewidth,page=1]{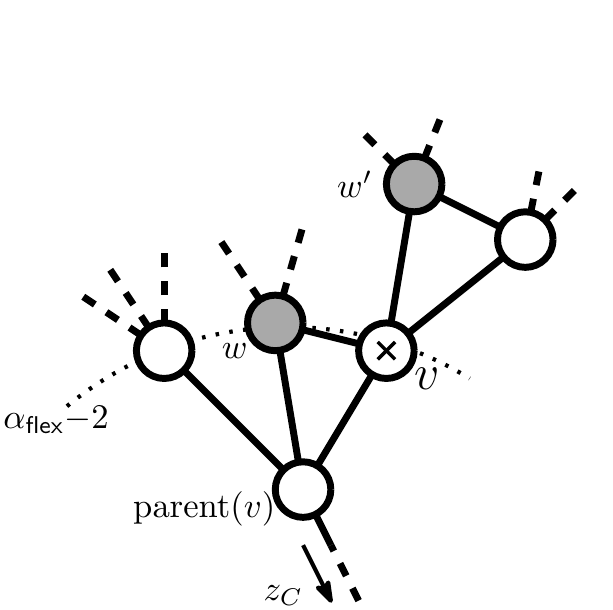}
\hspace{0.04\linewidth}\includegraphics[width=0.3\linewidth,page=2]{gatherers_v2}
\hspace{0.04\linewidth}\includegraphics[width=0.3\linewidth,page=3]{gatherers_v2}
\caption{The $3$ cases for the selection of independent gatherers: a node $v$ of depth $\aflex-2$ with at least one neighbor of equal depth (left); a node $v$ of depth $\aflex-2$ with two or more unconnected children (middle); a child $v'$ of a node $v$ of depth $\aflex-2$, where $v'$ has a neighbor of equal depth as $v$ had in the first case (right). In each case, gray nodes show two unconnected neighbors of the selected gatherer, marked by a cross.} 
\label{fig:gatherer-selection}
\end{figure}

\begin{claim}
\label{claim:number-gatherers}
For each cluster $C \in \Cflex$, either $C$ is adjacent to $\Vdcc$ ($N(C) \cap \Vdcc \neq \emptyset$), or $C$ contains at least $((\Delta-1)^{\floor{\aflex/2}-2})$ independent gatherers each with $(\Delta-1)^3$ or less conflicts. 
\end{claim}
\begin{proof}
Let us consider a cluster $C \in \Cflex$, $z_C$ its center, and let us assume that $N(C) \cap \Vdcc = \emptyset$. Recall that a node $v$ can be chosen as a potential \emph{gatherer} for $C$ if there exists $w,w' \in T_v$ s.t.\ $ww' \not \in E$, with $T_v = (N(v) \cap C) \setminus \set{\parent(v)}$. Two gatherers $v,v'$ are independent iff $(T_v \sqcup D_v) \cap (T_{v'} \sqcup D_{v'}) = \emptyset$, where $D_v$ is the subset of nodes of the cluster $C$ disconnected (in $C$) from $z_C$ by the removal of $T_v$.

We now want to measure how many independent gatherers can be selected in the cluster $C$.
Consider a BFS tree rooted at $z_C$, of depth $\aflex$, thus spanning the entire cluster. Note that this BFS tree is unique (\cref{lem:uniqueBFS}).
For each node $v\in V(C)$, let $d_v$ be its depth in the BFS (with $d_{z_C} = 0$). A node $v \neq z_C$ in $C$ has one parent node of depth $d_v - 1$, and $\Delta-1$ other neighbors of depth either $d_v$ or $d_v+1$ which form the set $T_v$. Any node $v$ such that $\dist(v,z_C) < \bflex$ and $\exists w \in N(v), d_w = d_v$ can be taken as gatherer, as $v$ has at least one child $w'$ in the BFS ($w'\in N(v)$, $d_{w'} = d_v+1$), which is not connected to $w$. Note that while such a $v$ is a gatherer, it is not independent from its up to $\Delta - 1$ neighbors of equal depth, but it is independent from all other gatherers of equal depth. Therefore, if we consider all the nodes at a given distance $r < \aflex$ from $z_C$, and focus on those with at least one neighbor at the same depth, at least a $1/(\Delta-1)$-fraction of them can be selected as independent gatherers.

When $v$ has no neighbor of equal depth, it can still be taken as gatherer if its $\Delta-1$ neighbors in $T_v$, which are all children of $v$, do not form a clique.
If they do, then any one of them can be taken as a gatherer, as long as they have depth $< \aflex$, i.e., $d_v \leq \aflex - 2$.

Considering all nodes $S(z_C,\aflex -2)$ at depth $\aflex -2$ in the BFS tree rooted at $z_C$, we can therefore  select at least $\card{S(z_C,\aflex -2)} / (\Delta-1)$ nodes of depth either $\aflex-2$ or $\aflex-1$ as independent gatherers. As the selected nodes are at depth $\aflex-2$ or higher, they have at most $(\Delta-1)^2$ descendants in the BFS tree at the ``border'' of the flex cluster, and as a result, conflict with less than $(\Delta-1)^3$ link clusters. \Cref{claim:no-dcc-expand} gives the necessary lower bound on $\card{S(z_C,\aflex -2)}$ yielding the claimed number of gatherers.
\end{proof}

\subsection{Properties of the Cluster Partition and its Cluster Hypergraph}
\label{sec:hypergraph-properties}

Recall the definition of the cluster hypergraph of a given cluster partition.

\DefClusterHypergraph*

In this section, we prove \cref{lem:cluster-partition-computing}, putting together the intermediate claims of previous sections. More precisely, we show that with the right choice of constants $\adcc$, $\aflex$, $\bflex$, $\alink$, and $\blink$, \BuildClusterPartition\ (\cref{alg:build-cluster-partition}) computes in $O(\Tmis(n,\poly(\Delta)))$ rounds a cluster partition with the properties claimed in \cref{lem:cluster-partition-computing}.

\LemComputingClusterPartition*

\begin{proof}
    The rank of the hypergraph is bounded by the maximum number of conflicts that each gatherer has, plus one. By \cref{claim:number-gatherers}, each gatherer has at most $(\Delta-1)^3$ conflicts, resulting in the $1+(\Delta-1)^3$ bound on the rank.

    Flex and link clusters that are adjacent to a DCC cluster get a rank $1$ hyperedge, and as such, do not need a lower bound on their degree. Similarly, a link cluster $C$ adjacent to a flex cluster whose gatherers are such that none of them conflicts with $C$ gets an hyperedge of rank $1$. It remains to lower bound the degree of flex and links clusters that do not get a rank $1$ hyperedge.
    
    The degree of a flex cluster $C$ without a rank $1$ hyperedge corresponds to the size $\card{\gath(C)}$ of its set of independent gatherers. This number is directly dictated by the choice of $\aflex$, by \cref{claim:number-gatherers}. We will fix $\aflex$ with this constraint in mind.
    
    The degree of a link cluster $C$ without a rank $1$ hyperedge is at least a
    $1/(\Delta-1)^3$-fraction of the number of edges it has to adjacent flex clusters, as this is the maximum number of edges that can be incident on the set of nodes in a gatherer's subtree at the border of the flex cluster. The number of edges between $C$ and neighboring flex clusters is driven by the choice of $\alink$ by \cref{claim:expansion-Clink}. We will fix $\alink$ with this constraint in mind.
    
    Finally, there are constraints connecting the various constants to one another. Notably, $\adcc$ needs to be large enough compared to the other constants for many arguments to go through, as we frequently make use of the absence of a DCC within some constant distance of node.
    Let us now set $\adcc,\aflex,\bflex,\alink,\blink$ as we need, mindful of the following constraints:
    \begin{itemize}
        \item $\aflex \geq 12$ and $\adcc \geq \aflex$ to find a minimum of $(\Delta-1)^{4}$ independent gatherers in flex clusters (\cref{claim:expansion-Cflex,claim:number-gatherers}),
        \item $\alink \geq 16$, $\bflex \geq 2(\aflex + \alink + 1)$, $\blink \geq 2\alink + (3/2)\bflex$, and $\adcc \geq 2(\alink+\bflex)$ to guarantee a minimum degree of $(\Delta-1)^{4}$ for link clusters (\cref{claim:expansion-Clink}),
    \end{itemize}

    The following parameters satisfy these constraints:
    \begin{itemize}
        \item $\adcc = 148$,
        \item $\aflex = 12$, $\bflex = 58$,
        \item $\alink = 16$, $\blink = 119$.
    \end{itemize}
    The different subroutines of \BuildClusterPartition\ (\cref{alg:build-cluster-partition}) are each performed in $O(\Tmis(n,\poly(\Delta)))$ rounds of LOCAL:
    \begin{itemize}
        \item $O(\adcc\cdot \Tmis(n,\Delta^{2\adcc+2})))$ for \BuildCdcc\ (\cref{alg:build-Cdcc,lem:build-Cdcc-result}),
        \item $O(\bflex\cdot \Tmis(n,\Delta^{\bflex})))$ for \BuildCflex\ (\cref{alg:build-Cflex}),
        \item $O(\blink\cdot \Tmis(n,\Delta^{\blink})))$ for \BuildClink\ (\cref{alg:build-Clink}).
    \end{itemize}
    The leading term is $O(\adcc \cdot \Tmis(n,\Delta^{\adcc})) = O(\Tmis(n,\Delta^{148}))$.
\end{proof}

Note that the minimum degree $\delta$ of the HSO instance can be increased to a higher $\poly(\Delta)$ by increasing the above constants to higher constant values, without affecting the rank of the hypergraph. We will tweak the constants to obtain a higher minimum degree $\delta$ to obtain our results in the randomized setting (\cref{sec:randomized-local}).

\subsection{Sinkless Orientation on the Cluster Hypergraph}
\label{sec:hso-on-cluster-hypergraph}

In this section, we prove \cref{lem:cluster-hypergraph-hso}, on the fact that a hypergraph sinkless orientation of the cluster hypergraph can be turned into useful sets $\fS$ and $I$ for $\Delta$-coloring. From this lemma and previously proved results directly follows \cref{lem:computing-DCC-T-nodes}, on how efficiently we can compute sets $\fS$ and $I$ in a complexity measured in terms of $\Tmis$ and $\Thso$.

We recall \cref{lem:cluster-hypergraph-hso} for convenience.

\LemClusterHypergraphOrientation*

\begin{proof}
    Let us use the bipartite hypergraph $H$ from \cref{def:cluster-hypergraph} with $\gath(C)$ the set of gatherers. 
    We build $\fS$ from the clusters in $\Cdcc$. Each cluster in $\Cdcc$ contains some induced subgraph $S_C$ that is either a DCC or a node of degree less than $\Delta$, and the subgraphs $(S_C)_{C \in \Cdcc}$
    are also non-adjacent. Therefore, $\fS = \set{S_C \mid C \in \Cdcc}$ is exactly as we need.

    We now turn to building the set $I$. For this, we will use $H$ and its hypergraph sinkless orientation.
    Each cluster $C \in \Cflex$ has at least one outgoing hyperedge according to this orientation.
    An hyperedge $e$ can be of two types. Either $e$ comes from the fact that the cluster is adjacent to a DCC cluster, or the hyperedge corresponds to a gatherer node. If the edge $e$ is of the second kind and an outgoing edge for $C$, then the cluster adds two nodes to $I$ because of that edge: two neighbors $w,w' \in T_v$ of the gatherer node $v$ that created the edge $e$ such that $w$ and $w'$ are not connected. Doing so, the gatherer $v$ joins the set $T_I$ of nodes with two neighbors in $I$.

    Consider the distance of nodes in $V\setminus I$ to the sets $\fS$ and $T_I$ in $G[V \setminus I]$, where $T_I$ is the set of nodes with two neighbors in $I$.
    Each node in $\Vdcc$ is at distance at most $\gdcc$ from $\fS$.
    Each node in a link or flex cluster adjacent to a cluster from $\Cdcc$ is at distance at most $\gdcc + \max(\gflex,\glink)$ from $\fS$.
    
    Consider now a node in a link cluster $C$ that is not adjacent to a DCC cluster. This cluster has a path in $G[V \setminus I]$ to the center of at least one adjacent flex cluster $C'$.
    If that flex cluster is adjacent to a cluster from $\Cdcc$, nodes in the link cluster $C$ are at distance at most $\gdcc + 2\gflex + \glink$ from $\fS$. If the flex cluster used one of its gatherers to add two of its nodes to $I$, then nodes in the link cluster $C$ are at distance at most $2\gflex + \glink$ to a node in $T_I \cap V(C')$.
    
    Finally, a node in a flex cluster $C$ that added two of its nodes to $I$ is at distance at most $2 \gflex$ to a node in $T_I$ if its path to $z_C$ the center of $C$ was not cut by $I$. A node in $C$ that was disconnected from the cluster's center $z_C$ by the cluster's contribution to $I$ still has a path of length at most $\bflex$ to a neighboring link cluster, and from this link cluster, can reach $\fS$ or $T_I$ in at most $\gdcc + 2\gflex + \glink$ additional steps.
    
    All in all, all nodes in $V\setminus I$ are at distance at most $O(\gflex+\glink+\gdcc)$ from either $\fS$ or $T_I$ in the graph $G[V \setminus I]$.
\end{proof}

\LemComputingDCCandTnodes*

\begin{proof}
By \cref{lem:cluster-partition-computing}, in $O(\Tmis(n,\poly(\Delta)))$ rounds, we can compute a cluster partition and a set of gatherers as in the hypotheses of \cref{lem:cluster-hypergraph-hso}.
Only missing is the hypergraph sinkless orientation.
To compute it, we first have all clusters with a rank $1$ hyperedge grab it, i.e., all such clusters get a trivial outgoing edge.
Then, for the remaining clusters, we compute a hypergraph sinkless orientation in $O(\Thso(n,\poly(\Delta)))$ rounds, since \cref{lem:cluster-partition-computing} guarantees a degree-to-rank ratio of at least $(\Delta-1)^4 / ((\Delta-1)^3 + 1) > \sqrt{\Delta}$, and a round of communication on the hypergraph can be simulated in $O(1)$ rounds of communication on $G$.
Applying \cref{lem:cluster-hypergraph-hso}, we get sets $\fS$ and $I$ as claimed in the lemma.
\end{proof}

\subsection{Obtaining a \texorpdfstring{$\Delta$}{∆}-Coloring}
\label{sec:finish-coloring}

We now show how to get a $\Delta$-coloring of the nodes of the graph, using what we have shown so far. 

\ThmLOCALmain*

As final ingredient before proving \cref{thm:local-main}, we prove \cref{lem:finish-coloring}.

\LemFinishColoring*

\begin{proof}
Consider the set $\fS$ and $I$. Each node in $I$ is given color $1$. Then, we compute the distance of each node to $\fS \sqcup T_I$ in $G[V \setminus I]$.
This distance is at most $d$ by assumption, and can therefore be computed in $O(d)$ rounds of \LOCAL.
As each node until the last layer has at least one neighbor in a layer of smaller depth, coloring each layer is an instance of degree+1-list coloring, and coloring all the nodes is an instance of coloring a layered graph with $d$ layers and maximum up-degree $\Delta-1$.
Once all layers except that of depth $0$, i.e., nodes in $\fS \sqcup T_I$ have been colored, extending the coloring to $\fS \sqcup T_I$ is trivial, as the connected components induced by $\fS \sqcup T_I$ have constant weak diameter and can be colored with colors from $[\Delta]$ (either from degree-choosability, from the lower degree of a node, or adjacency to two neighbors of color $1$).
\end{proof}

As explained in the preliminaries (\cref{sec:prelims}), note that coloring a layered graph with $d$ layers simply reduces to solving $d$ instances of degree+1-list-coloring.
We mention this structure of the layered graph instead of stating that it suffices to solve $d$ instances of degree+1-list-coloring in succession because some nontrivial improvements over this basic strategy have been shown to sometimes be possible in a prior work~\cite{ghaffari2021deterministic}.

\begin{proof}[Proof of \cref{thm:local-main}]
The algorithm first start by computing a $(\gdcc,\gflex,\glink)-\Delta$-coloring cluster partition as in \cref{lem:cluster-partition-computing} in $O(\Tmis(n,\poly(\Delta)))$ rounds.
The bipartite cluster multihypergraph from \cref{def:cluster-hypergraph} that we obtain has a maximum rank of $(\Delta-1)^3+1$ while each of its nodes either has an incident hyperedge of rank $1$ or a degree of at least $(\Delta-1)^4$.

As shown in \cref{lem:computing-DCC-T-nodes}, this allows us to compute an HSO of $H$ in $O(\Thso(n,\poly(\Delta)))$ rounds using \cref{thm:HSO}, and the resulting HSO can be turned into sets $\fS$ and $I$ for use in \cref{lem:finish-coloring}, notably, such that every node in $V \setminus I$ is at distance at most $d = O(1)$ from a node in $\fS \sqcup I$ in the graph $G[V \setminus I]$. Also, each member of $\fS$ has weak diameter $k = O(1)$.

This last coloring step takes $\Tlayer(n,\Delta-1,d)+O(k)$ rounds, and results in a $\Delta$-coloring of the graph. Since $\Tlayer(n,\Delta-1,d) \in O(d \cdot \Tmis(n,\poly(\Delta)))$, $d \in O(1)$, and $k \in O(1)$ the complexity of this last step does not exceed $O(\Tmis(n,\poly(\Delta)))$.
In total, we have computed a $\Delta$-coloring in $O(\Tmis(n,\poly(\Delta)) + \Thso(n,\poly(\Delta)))$ rounds.
\end{proof}

\subsection{Randomized LOCAL Setting}
\label{sec:randomized-local}

In this section, we show how to obtain the randomized version of \cref{thm:local-main}, our deterministic result.
We use the same chain of reductions as in the deterministic, but simply use of randomized algorithms as subroutines to solve our instances of MIS and HSO.
As known results for solving HSO in the randomized setting require some additional constraints on the minimum degree $\delta$ and the rank $r$ of the hypergraph (see \cref{thm:HSO_randomized}), we have to slightly adjust the constants $\adcc$, $\aflex$, $\bflex$, $\alink$, and $\blink$ that we use in the construction of our cluster partition.

\ThmLOCALmainrand*

\begin{proof}[Proof of \cref{thm:local-main-rand}]
The algorithm once again starts by computing a cluster partition as in \cref{lem:cluster-partition-computing}, running \BuildClusterPartition\ (\cref{alg:build-cluster-partition}) with appropriate values for $\adcc$, $\aflex$, $\bflex$, $\alink$ and $\blink$.
In addition to the constraints on these constants explained in \cref{sec:hypergraph-properties}, in order to solve HSO on the cluster hypergraph, we additionally require that $\delta > 320 r \log r$, where $\delta$ is the maximum degree of the graph and $r$ the rank.
With a rank that remains at $r = 1+ (\Delta-1)^3$, and as $\Delta \geq 3$, we can achieve this by increasing the minimum degree $\delta$ to $(\Delta-1)^{14}$ since 
\begin{align*}
320r \log r = 320 ((\Delta-1)^3+1) \log ((\Delta -1)^3 +1) 
 < 2^{10} (\Delta-1)^4 \leq (\Delta-1)^{14}\ .
\end{align*}
To obtain a minimum degree this large, we take the following parameters:
\begin{itemize}
    \item $\adcc = 348$,
    \item $\aflex = 32$, $\bflex = 138$,
    \item $\alink = 36$, $\blink = 279$.
\end{itemize}

After computing the cluster partition, the rest of the algorithm is exactly the same as in the deterministic setting.
Computing the cluster partition was done in
$O(\Tmisrand(n,\poly(\Delta)))$ rounds. 
Then $O(\Thsorand(n,\poly(\Delta),\poly(\Delta)))$ rounds are devoted to computing a sinkless orientation on the obtained bipartite cluster multihypergraph (\cref{def:cluster-hypergraph}).
Following this, \cref{lem:cluster-hypergraph-hso} gives the set $\fS$ and $I$ which are then used to compute the coloring in $\Tlayerrand(n,\Delta-1,d)+O(k) \leq O(\Tmisrand(n,\poly(\Delta)))$ rounds (\cref{lem:finish-coloring}). The overall round-complexity is $O(\Tmisrand(n,\poly(\Delta)) + \Thsorand(n,\poly(\Delta),\poly(\Delta)))$.
\end{proof}

\section{Bounded Clique Number}

In this section, we prove \cref{thm:local-clique-number}, our result for graphs of bounded clique number $\omega$.

\ThmLOCALcliqueNum*

The algorithm behind this result is quite simple, only consisting of three steps.
First, \cref{alg:clique-num} computes an initial $\Delta+1$-coloring of the graph.
Second, the algorithm focuses on nodes colored with one of $\omega+1$ colors. After a simple greedy process taking $O(\omega)$ rounds (\cref{alg:color-reduction}), the nodes that are still colored with one of the $\omega+1$ colors are guaranteed to induce a subgraph $H$ of maximum degree $\omega$.
As the computation takes place on a $K_{\omega+1}$-free graph $G$ by assumption, the subgraph $H$ is also $K_{\omega+1}$-free, and therefore, $\omega$-colorable.
$\omega$-coloring $H$ saves one color, resulting in a $\Delta$-coloring of the graph $G$.
To $\omega$-color $H$, we use our results for general graphs on $H$. In effect, we reduce $\Delta$-coloring on $K_{\omega+1}$-free graphs to an instance of $\Delta$-coloring on general graphs, but with a smaller maximum degree of $\omega\geq 3$.

\begin{algorithm}[htb!]
\caption{Algorithm for $\Delta$-coloring graphs of bounded clique number $\omega\geq 3$}
\label{alg:clique-num}
\begin{algorithmic}[1]
\State Compute a $\Delta+1$ coloring.
\State \GreedyColorReduction\ with $C=[1,\omega+1]$, $C' = [\Delta+1] \setminus C$.
\Statex Let $H$ be the subgraph induced by the nodes colored by a member of $[1,\omega+1]$.
\State $\omega$-color $H$
\end{algorithmic}
\end{algorithm}

\begin{algorithm}[htb!]
\caption{\GreedyColorReduction}
\label{alg:color-reduction}
\begin{algorithmic}[1]
\Statex \textbf{Input:} set of marked colors $C$, set of replacement colors $C'$, colored graph $G = (V,E)$.
\Statex Let $\psi(v)$ denote the current color of a node $v$.
\For{each $c \in C$, sequentially}
    \For{each node $v \in V$ s.t.\ $\psi(v) = c$, in parallel}
    \If{$\exists c' \in C'$ s.t.\ $\forall v' \in N(v), \psi(v') \neq c'$}
        \State $\psi(v) \gets c'$.
    \EndIf
    \EndFor
\EndFor
\end{algorithmic}
\end{algorithm}

\begin{claim}
\label{claim:color-reduction}
    Suppose we are given a $(\Delta+1)$-colored graph $G$. After running \GreedyColorReduction (\cref{alg:color-reduction}) on $G$ with a set $C$ of $\card{C} = k$ marked colors, and $C'=\set{1,\dots,\Delta+1} \setminus C$ as replacement colors, the graph induced by nodes colored with a member of $C$ has maximum degree $k-1$.
\end{claim}

\begin{proof}
Let $H = (V_H,E_H)$ be the graph induced by nodes colored with a color from $C$ at the end of the algorithm.
    During the algorithm, nodes can only update their color to a color in $C'$. Therefore, if a node has a marked color $c\in C$ at the end of the algorithm, it had this color for the entire duration of the algorithm.
    
    Consider a node $v \in V_H$ of color $c\in C$ at the end of the algorithm. Since $v$ is colored with a marked color at the end of the algorithm, $v$ was not able to take a color from the replacement colors $C'$ when the for-loop over the marked colors in \GreedyColorReduction\ reached color $c$.
    Hence, when it had the opportunity to change its color, $v$ had at least $\card{C'} \geq \Delta + 1 - k$ neighbors colored with a replacement color. It therefore had at most $\Delta - \card{C'} \leq k-1$ neighbors colored with a color from $C$ at that point. Since the number of neighbors with a marked color can only decrease as \GreedyColorReduction\ runs, $v$ still has at most $k-1$ neighbors with a marked color at the end of the algorithm. I.e., $H$ has maximum degree $k-1$.
\end{proof}

\begin{proof}[Proof of \cref{thm:local-clique-number}]
    \Cref{alg:clique-num} starts by computing a $\Delta+1$ coloring algorithm, which takes time $\Tdelpoc(n,\Delta)$ ($\Tdelpocrand(n,\Delta)$ in randomized \LOCAL). Then, it runs $\GreedyColorReduction$ with a set of $\omega+1$ colors, which takes $O(\omega)$ rounds.
    After this step, the induced subgraph $H$ of $G$ obtained by considering only the nodes with color between $1$ and $\omega+1$ has a maximum degree of $\omega$, by \cref{claim:color-reduction}.
    $H$ is a collection of connected $H_1,\dots,H_l$ such that each $H_i, i\in [l]$ is not an $(\omega+1)$-clique, since $G$ is $K_{\omega+1}$-free. This allows us to apply our $\Delta$-coloring algorithm \cref{alg:main-algorithm} to $H$ and obtain a $\omega$-coloring of $H$ in $\Tdelc(n,\omega)$ rounds ($\Tdelcrand(n,\omega)$ in randomized \LOCAL). This eliminates one of the $\omega+1$ colors, which means that the graph $G$ is colored with $\Delta$ colors.
\end{proof}

\paragraph{On the minimum value of $\omega$ and triangle-free graphs}
We note that attempting to perform the same reduction with $\omega=2$ would result in an induced subgraph $H$ of maximum degree $2$. Such a graph $H$ could contain odd cycles, and therefore not be $\omega$-colorable. A lower bound of $\omega \geq 3$ ensures that the induced subgraph $H$ is not only $\omega$-colorable, but efficiently so.

The lower bound on $\omega \geq 3$ should not be understood as limiting the range of application of \cref{thm:local-clique-number} to not include triangle-free graphs ($K_{\omega+1}$-free graphs with $\omega = 2$). $K_3$-free graphs are a fortiori $K_4$-free, so the theorem can be applied with $\omega = 3$.

\section*{Acknowledgements}
Funded by the European Union. Views and opinions expressed are however those of the author(s) only and do not necessarily reflect those of the European Union or the European Research Council. Neither the European Union nor the granting authority can be held responsible for them. This work is supported by ERC grant \href{https://doi.org/10.3030/101162747}{OLA-TOPSENS} (grant agreement number 101162747) under the Horizon Europe funding programme.

\urlstyle{same}
\bibliographystyle{alpha}
\bibliography{0a_arxiv_v2_refs_preamble,0b_arxiv_v2_refs_entries}

\appendix

\end{document}